\documentclass[format=acmsmall, review=false]{acmart}
\usepackage{acm-ec-25}
\usepackage{booktabs} 
\usepackage[ruled]{algorithm2e} 
\usepackage{multicol}
\usepackage{multirow}

\SetAlFnt{\small}
\SetAlCapFnt{\small}
\SetAlCapNameFnt{\small}
\SetAlCapHSkip{0pt}
\IncMargin{-\parindent}

\usepackage{algorithmic}
\usepackage{amsmath}
\usepackage{amsthm}
\usepackage{subcaption}
\usepackage{booktabs}
\usepackage{xcolor}

\setcitestyle{authoryear}

\title[The Art of Two-Round Voting]{The Art of Two-Round Voting}

\author{Qishen Han}
\email{hnickc2017@gmail.com}
\orcid{0000-0003-0268-6918}
\affiliation{%
  \institution{Rutgers University}\city{Piscataway}\state{NJ}\country{USA}}
\author{Grant Schoenebeck}
\email{schoeneb@umich.edu}
\orcid{0000-0001-6878-0670}
\affiliation{%
  \institution{University of Michigan}\city{Ann Arbor}\state{MI}\country{USA}}
\author{Biaoshuai Tao}
\email{bstao@sjtu.edu.cn}
\orcid{0000-0003-4098-844X}
\affiliation{%
  \institution{Shanghai Jiao Tong University}\city{Beijing}\country{China}}
\author{Lirong Xia}
\email{xialirong@gmail.com}
\orcid{0000-0002-9800-6691}
\affiliation{%
  \institution{Rutgers University and DIMACS}\city{Piscataway}\state{NJ}\country{USA}}

\begin{abstract}
We study the voting problem with two alternatives where voters' preferences depend on a not-directly-observable state variable. While equilibria in the one-round voting mechanisms lead to a good decision, they are usually hard to compute and follow. 
We consider the two-round voting mechanism where the first round serves as a polling stage and the winning alternative only depends on the outcome of the second round.
We show that the two-round voting mechanism is a powerful tool for making collective decisions. Firstly, every (approximated) equilibrium in the two-round voting mechanisms (asymptotically) leads to the decision preferred by the majority as if the state of the world were revealed to the voters. Moreover, there exist natural equilibria in the two-round game following intuitive behaviors such as informative voting, sincere voting~\citep{Austen96:Information}, and surprisingly popular strategy~\citep{prelec2017solution}. 
This sharply contrasts the one-round voting mechanisms in the previous literature, where no simple equilibrium is known. Finally, we show that every equilibrium in the standard one-round majority vote mechanism gives an equilibrium in the two-round mechanisms that is not more complicated than the one-round equilibrium. Therefore, the two-round voting mechanism provides a natural equilibrium in every instance including those where one-round voting fails to have a natural solution, and it can reach an informed majority decision whenever one-round voting can. Our experiments on generative AI voters also imply that two-round voting leads to the correct outcome more often than one-round voting under some circumstances. 

\end{abstract}

\begin{document}

\newtheorem{thm}{Theorem}
\newenvironment{thmnb}[1]
  {\renewcommand{\thethm}{\ref{#1}'}%
   \addtocounter{thm}{-1}%
   \begin{thm}}
  {\end{thm}}

  \newenvironment{lembis}[1]
  {\renewcommand{\thelem}{\ref{#1}}%
   \addtocounter{lem}{-1}%
   \begin{lem}}
  {\end{lem}}

\newtheorem{claim}{claim}
\newtheorem{prop}{Proposition}
\newenvironment{propnb}[1]
  {\renewcommand{\theprop}{\ref{#1}'}%
   \addtocounter{prop}{-1}%
   \begin{prop}}
  {\end{prop}}
\newtheorem{lem}{Lemma}
\newenvironment{lemnb}[1]
  {\renewcommand{\thelem}{\ref{#1}'}%
   \addtocounter{lem}{-1}%
   \begin{lem}}
  {\end{lem}}

\newtheorem{asm}{Assumption}
\newtheorem{conj}{Conjecture}

\newtheorem{Alg}{Algorithm}
\newtheorem{prob}{Problem}
\newtheorem{prof}{Proof}
\newtheorem{coro}{Corollary}
\newenvironment{coronb}[1]
  {\renewcommand{\thecoro}{\ref{#1}'}%
   \addtocounter{coro}{-1}%
   \begin{coro}}
  {\end{coro}}
  
\theoremstyle{acmdefinition}
\newtheorem{dfn}{Definition}
\newtheorem{ex}{Example}

\newcommand\qishen[1]{{\color{blue} \footnote{\color{blue}Qishen: #1}} }
\newcommand\gs[1]{{\color{orange} \footnote{\color{orange}Grant: #1}} }

\newcommand{\blue}[1]{\textcolor{blue}{#1}}

\newcommand{\sag}{i}
\newcommand{\ag}{n}
\newcommand{\pr}{p}
\newcommand{\prnum}{m}
\newcommand{\prvec}{M}
\newcommand{\prprf}{\mathcal{M}}
\newcommand{\func}{f}
\newcommand{\prdist}{\Pi}
\newcommand{\rpvec}{\hat{M}}
\newcommand{\rpprf}{\hat{\mathcal{M}}}
\newcommand{\jdg}{\mathcal{J}}
\newcommand{\ut}{u}
\newcommand{\vt}{v}
\newcommand{\Ut}{U}
\newcommand{\stg}{\sigma}
\newcommand{\Stg}{\Sigma}
\newcommand{\stgp}{\Sigma}

\newcommand{\acc}{A}
\newcommand{\err}{I}

\newcommand{\atn}{j}
\newcommand{\Atn}{J}
\newcommand{\Atnset}{\mathcal{A}}
\newcommand{\wos}{k} 
\newcommand{\Wos}{K} 
\newcommand{\Wosset}{\mathcal{W}} 
\newcommand{\Wosrv}{W} 
\newcommand{\sig}{m} 
\newcommand{\sigi}{s}
\newcommand{\Sig}{M} 
\newcommand{\Sigset}{\mathcal{S}} 
\newcommand{\Sigrv}{S} 
\newcommand{\Sigp}{Sp}
\newcommand{\Siggroup}{Sg}
\newcommand{\wn}{r}

\newcommand{\bA}{\mathbf{A}}
\newcommand{\bR}{\mathbf{R}}
\newcommand{\bP}{\mathbf{P}}

\newcommand{\hist}{H}
\newcommand{\Histset}{\mathcal{H}}

\newcommand{\agh}{\alpha_H}
\newcommand{\ags}{\alpha_S}
\newcommand{\aga}{\alpha_A}
\newcommand{\Thd}{\mu}

\newcommand{\hst}{x} 
\newcommand{\hstset}{X}
\newcommand{\termhst}{Z}
\newcommand{\rhst}{\hst_{0}}
\newcommand{\thst}{z}
\newcommand{\blf}{\mu} 
\newcommand{\blfp}{M}
\newcommand{\typevec}{t}
\newcommand{\typevecset}{\mathcal{T}}
\newcommand{\infpart}{\mathcal{I}}
\newcommand{\infset}{I}
\newcommand{\oneround}{k}
\newcommand{\plf}{Pl}

\newcommand{\cst}{c} 
\newcommand{\ct}{c}
\newcommand{\spt}{th} 

\newcommand{\Rdone}{Rd^{1}}

\newcommand{\agf}{\alpha_{F}}
\newcommand{\agu}{\alpha_{U}}
\newcommand{\agc}{\alpha_{C}}
\newcommand{\tf}{\ag_{F}}
\newcommand{\tu}{\ag_{U}}
\newcommand{\tc}{\ag_{C}}
\newcommand{\lp}{\lambda}
\newcommand{\instance}{\mathcal{I}}
\newcommand{\textinst}{environment}

\newcommand{\thh}{T}
\newcommand{\exs}{\varphi}

\newcommand{\stgbar}{\hat{\stg}}
\newcommand{\stgpbar}{\hat{\stgp}}

\newcommand{\stgone}{\stg^{1}}
\newcommand{\stgtwo}{\stg^{2}}
\newcommand{\aYES}{{\bf\color{blue}YES}}
\newcommand{\aNO}{{\bf\color{orange}NO}}

\newcommand{\Qishen}[1]{\textbf{\color{orange} [Qishen: #1\textbf{]}}} 

\newcommand{\nocontentsline}[3]{}
\newcommand{\tocless}[2]{\bgroup\let\addcontentsline=\nocontentsline#1{#2}\egroup}

\begin{titlepage}

\maketitle

\vspace{1cm}
\setcounter{tocdepth}{1} 
\tableofcontents

\end{titlepage}


\section{Introduction}
\label{sec:intro}
{
How can people make wise informed collective decisions?  In many scenarios, people's preferences depend on the information they receive, and a collective decision needs to be made fairly. 
For example, suppose that a community wants to decide whether or not to adopt a restrictive COVID-19 policy. The preferences of the community members depend on their perceptions of the virality of the pandemic. While the virality is not directly observable, community members receive (noisy) signals about it to form their preferences, and their preferences can still be different even upon receiving the same signal.  For example, immunocompromised people may still prefer to adopt the restrictive policy when the signal is ``moderately dangerous'', while healthy people may prefer not to adopt the policy after receiving the same signal. Other examples include political elections, recruiting, group activity selection, etc.

A natural idea is to use voting to make collective decisions. When there are two alternatives, majority voting appears to be a natural choice. Majority voting reveals the state of the world (the virality of COVID-19 in the above example) almost surely when the population is large, according to the Condorcet Jury theorem~\citep{Condorcet1785:Essai}. 
Moreover, when voters are strategic, majority voting can still aggregate voters' fully informed preferences via strategic behaviors~\citep{Feddersen97:Voting, han2023wisdom}. In other words, it chooses the majority winner {\em as if the state of the world were revealed to the voters.} In the literature, such majority decision has been referred to as \emph{full information equivalence}~\citep{Feddersen97:Voting}, \emph{majority wish}~\citep{schoenebeck21wisdom}, and \emph{informed majority decision}~\citep{han2023wisdom}.
In this paper, we will use the term ``informed majority decision''. 


However, all equilibria of such one-round majority games in previous work are quite complicated and hard to follow. For example, the seminal paper by~\citet{Austen96:Information} shows that, under majority voting, the natural and intuitive behavior of {\em informative voting}, in which every voter votes for his/her signal, and {\em sincere voting}, in which every voter votes as if he/she is making an individual decision and vote for his/her preferred alternative in expectation, are not Nash equilibria in most situations.
The equilibrium constructed by~\citet{Feddersen97:Voting} requires the computation of a complicated Riemann integral. The equilibrium constructed by~\citet{han2023wisdom} involves carefully calculated mixed strategies. 


In analyzing strategic behavior, we emphasize the importance of natural and intuitive equilibria. It is quite unlikely that real-world agents can calculate out and play to counter-intuitive mixed or asymmetric strategies, let alone play the same equilibrium. Therefore, these counter-intuitive strategy profiles fail to predict the outcome of strategy agents as they are intended for. Moreover, the simplicity of the mechanism is also a desirable property and a proxy of usability, as the theoretical guarantees of complex mechanisms may not translate into practice.


\begin{center}
{\bf Can we design {\em simple} mechanisms with {\em natural} equilibria to achieve the informed majority decision?}
\end{center}

The definition of a natural equilibrium depends on the scenarios and the purposes. In this paper, we identify the following criteria. Firstly, agents play pure/deterministic strategies. Mixed or randomized strategies are hard for agents to compute and implement and are rarely applied by real-world voters. Secondly, agents play symmetrically, i.e., agents with the same preferences should play the same strategy. Otherwise, the mapping from similar agents to diverse strategies is unclear and unincentivized, and the equilibrium is unlikely to predict the outcome. Moreover, agents can easily compute the strategies. Two examples of intuitive and natural strategies are informative voting and sincere voting, where the first simply follows the signal, and the second follows the preferences in expectation. 

Unfortunately, previous one-round voting mechanisms fail to have a natural equilibrium. 
Moreover, not only is a natural equilibrium not identified in these papers, but \textbf{one round voting game with the majority rule has NO pure symmetric equilibrium} under a large class of games (Proposition ~\ref{prop:unnatural}). The known equilibria are either randomized or asymmetric and are usually hard to compute. 


\begin{ex}[Failure of informative voting being a Nash equilibrium~\citep{Austen96:Information}]
Consider the aforementioned COVID-19 policy problem. Agents are deciding whether to accept or reject a restrictive policy. The world state, a.k.a. the virality of the pandemic, can be of High risk or Low risk. Agents receive High or Low signals related to the world state. The signal distribution under different risk levels and the utility of agents under different states and outcomes are specified in Table~\ref{tbl:failure}. The signal distribution is biased, as an agent is always more likely to receive a High signal no matter what the risk level is.


\begin{table}[htbp]
    \begin{subtable}[t]{0.5\linewidth}
        \centering
\begin{tabular}{@{}ccc@{}}
\toprule
State of the World    & High Signal & Low Signal \\ \midrule
High Risk & 0.9         & 0.1        \\
Low Risk  & 0.6         & 0.4        \\ \bottomrule
\end{tabular}
\subcaption{Signal distributions.}
    \end{subtable}
    \begin{subtable}[t]{0.4\linewidth}
        \centering
\begin{tabular}{@{}ccc@{}}
\toprule
State    & Accept & Reject  \\ \midrule
High Risk & 1         & 0       \\
Low Risk  & 0         & 1       \\ \bottomrule
\end{tabular}
\subcaption{Agents' utilities.}
    \end{subtable}
\caption{An instance where informative voting fails to be an equilibrium. \label{tbl:failure}}
\end{table}
\vspace{-1em}

A strategic agent considering a unilateral deviation only cares about the case where exactly half of the other agents vote for Accept and the other half vote for Reject because otherwise her vote will not affect the outcome. If all other agents vote informatively, this implies exactly half of the other agents receive High signals. However, agents receive High signals with a high probability of 0.9 when the risk level is High. Therefore, the world state is much more likely to be Low risk when half of the agents receive Low signals, and a strategic agent will deviate and vote for Reject even if her signal is High. This directly implies that informative voting fails to be a Nash equilibrium. $\hfill \qed$

\end{ex}

Another line of work designs mechanisms to achieve the informed majority decision via simple equilibria, yet the mechanism structures are sophisticated. \citet{schoenebeck21wisdom} propose an information-aggregation mechanism that incentivizes truthful reporting and reaches the informed majority decision with a high probability. However, the mechanism leverages the ``surprisingly popular'' mechanism~\citep{prelec2017solution} and the median-voter theorem~\citep{Black48:Rationale}, 
resulting in a structure that is quite sophisticated and may not be trusted by the people. Moreover, The mechanism elicits a distribution from each agents, which requires infinite bits to be represented. This further increases the difficulty of implementing the mechanism. 

\subsection{Our Contribution}

In this paper, we provide a solution to the research question of simple mechanism design---\emph{two-round voting}, which has a polling stage (first round) and a voting stage (second round). 
After the first round but before the second, the agents can observe the number of votes cast for each alternative in the first round, and the majority vote in the second round completely determines the winner. 

For a sharp contrast, we first show that the one-round voting game with the majority rule has NO pure and symmetric equilibrium under a large class of games (Proposition~\ref{prop:unnatural}), which also rules out the natural informative voting and sincere voting. We also illustrate the existing equilibria which either mixed or asymmetric, making them difficult for agents to calculate and coordinate. 

On the other hand, the two-round voting mechanism is a simple and powerful mechanism with natural equilibria.
Firstly, we prove that the two-round voting mechanism indeed achieves good decisions with strategic agents, as every $\varepsilon$-approximate 
strong Bayes Nash equilibrium, in which no group of agents has incentives to deviate, leads to the informed majority decision with high probability (Theorem~\ref{thm:eqtoacc}). The probability converges to 1 as $\varepsilon$ goes to 0. 
Secondly, the two-round voting mechanism has the merit of a simple mechanism in both structure-wise and equilibrium-wise. It has a simple structure that involves simply conducting a poll before the actual vote. Such a polling-voting structure has been widely applied in real-world political elections. 

Moreover, we show a natural 
$\varepsilon$-approximate strong Bayes Nash equilibrium that leads to the informed majority decision (Theorem~\ref{thm:sincere}).  As the number of voters increases, $\varepsilon$ converges to 0. The equilibrium combines the most natural and classical informative voting and sincere voting. In this equilibrium, agents whose preferred alternative changes with the world state vote informatively (vote for their signal) in the first round, and vote sincerely (vote for the alternative they preferred in expectation conditioned on the first-round outcome) in the second round. The ``informative+sincere'' equilibrium follows the intuition of sharing information and choosing the (expected) preferred alternative. 
We further discover that the ``informative+sincere'' equilibrium belongs to a class of ''informative + threshold'' equilibria (Theorem~\ref{thm:threshold}), in which agents vote informatively in the first round and vote for the alternative whose first-round vote exceeds a threshold. The ''informative + threshold'' class derives more natural equilibria, such as the ''informative + surprsingly popular'' equilibrium in which agents vote for the ``surprsingly popular'' alternative in the second round.

Furthermore, we show with mild assumptions that, every one-round equilibrium gives a two-round equilibrium that is not more ``complicated'' (Theorem~\ref{thm:one_and_two}). Specifically, agents in the two-round voting mechanism can ignore the polling stage and directly play the one-round equilibrium in the second round. 
Therefore, the two-round voting mechanism provides natural equilibrium in every instance including those where one-round voting fails to have a natural solution, and it can reach an informed majority decision whenever one-round voting can. 

We also test the one-round voting and the two-round voting with synthetic experiments where we apply generate AI agents to simulate voters. When the signal distribution is biased, the AI agents powered by the reasoning model reach the informed majority decision more often in the two-round voting than in the one-round voting. We also observe behaviors such as Bayesian update and vote switch (in the second round). Our experiments provide insights into future applications for AI-augmented or proxy votes in the future and serve as references for real-world experiments.

\subsection{Related Work and Discussion}
The study of the binary voting game with imperfectly informed agents originates from the famous Condorcet Jury Theorem~\citep{Condorcet1785:Essai}, which shows the strength of voting when seeking the truth. Condorcet showed that the majority vote can reveal the ``correct alternative'' with probability converging to one as the number of voters increases, conditioned on (1) agents honestly reflecting their information in the votes (informative voting), and (2) the probability of an agent getting the correct information $p > 0.5$. 

The seminal paper by~\citet{Austen96:Information} is the first to study this problem in a game-theoretical setting. They show that informative voting is not compatible with the strategic behavior of voters. As a consequence, a large literature focuses on the existence of non-informative equilibrium and its effect on making a good decision. \citet{Wit98:Rational}, \citet{Feddersen1998:innocent}, and \citet{Myerson98:Extended} discovers the existence of a mixed strategy equilibrium that reveals the ground truth under different settings. \citet{Feddersen97:Voting} propose a unique class of equilibria that reach the informed majority decision. \citet{han2023wisdom} reveal the equivalence between a strategy profile being an equilibrium and reaching a good decision. None of these works guarantee the existence of a ``simple'' equilibrium. 

Multiple works go ``beyond'' the one-round voting game, investigating and proposing more sophisticated mechanisms to aggregate information. \citet{schoenebeck21wisdom} design an information aggregation mechanism that incentivizes informative voting and leads to an informed majority decision with high probability. Their mechanism leverages the ``surprisingly popular mechanism'' and the median-voter theorem to incentivize informative votes. \citet{Coughlan2000:-In-Defense} shows that when agent preferences are restricted to similar, adding an communication phase before voting leads to a good decision. Their result does not extend to scenarios where agents have diverse preferences as in our setting because they adopted a different solution concept (Bayesian-Nash equilibrium) and their techniques are very different from ours. \citet{Morgan2008Polls} study a two-round decision-making game in which voters vote in the first-round poll and the individual decision maker chooses the decision in the second round. In their model, agents cannot make a direct decision but only affect the decision-maker by their reported signals.
\citet{Dekel00:Sequential} study the voting game in which agents publicly declare their vote in a predetermined order. They show that a strategy profile is an equilibrium in sequential voting if and only if a corresponding profile is an equilibrium in one-round voting. This implies that sequential voting may not have a simple equilibrium. Both papers are very different from ours in the model and the results.  
\citet{Chopra2004Knowledge} studies the necessary knowledge for agents to strategize in the voting problems. 
Other directions of one-round voting with imperfect information include  dependent signals~\citep{Nitzan84:Significance,Shapley84:Optimizing,Kaniovski10:Aggregation}, agents with different competencies~\citep{Nitzan80:Investment,Gradstein87:Organizational,Ben11:Condorcet}, voting with more than two alternatives~\citep{Young88:Condorcet,Goertz14:Condorcet}, and empirical studies~\citep{Ignacio2014hypo,Battaglini2010swing}. 

Our paper also aligns with the idea that communication in voting improves aggregation, especially~\citet{Geraldi2007deliberative} which shows that communication before voting helps stabilize the voting outcome. Their results show that equilibrium outcomes in “deliberation+different voting rules” are the same. \citet{Goeree2011experimental} further show empirically that these outcomes resemble those in one-round informative voting. They don’t examine whether outcomes lead to good decisions, while we show that all equilibria in the two-round voting lead to good decisions. 

Our paper is also related to strategic information transmission where information communication occurs between informed experts and an uninformed decision maker~\citep{crawford1982strategic,battaglini2004policy,Morgan2008Polls,austen1993interested,battaglini2002multiple}. Divergent preferences between experts and decision-makers usually lead to information loss. However, even when experts strategically report their information, the decision maker can still collect enough information to make an optimal decision~\citep{baharad2012beyond,Morgan2008Polls}.

}

\section{Preliminaries}
\label{sec:prelim}
{In this section, we introduce the information structure and the one-round voting game model adopted from previous work~\citep{schoenebeck21wisdom,han2023wisdom}. 

\paragraph{Alternatives and World States.} $\ag$ agents vote for two alternatives $\bA$ and $\bR$. Two possible world states $\Wosset = \{L, H\}$ reflect the overall tendency of the preferences, 
where $\bA$ is more preferred in $H$, and $\bR$ is more preferred in $L$. We use $k$ to denote a generic world state. The world state cannot be observed directly by the agents. Let $P_{H} = \Pr[\Wosrv = H]$ and $P_{L} = \Pr[\Wosrv = L]$ be the common prior of the world states. We assume $P_H > 0$ and $P_L > 0$.

\paragraph{Private Signals.} Each agent receives a signal in $\Sigset = \{l, h\}$ related to the world state. We use $\sig$ to denote a generic signal and $\Sigrv_\sag$ to denote the random variable representing the signal that agent $\sag$ receives. We assume the signals are independent and have identical distributions conditioned on the world state. $P_{\sig\wos} = \Pr[\Sigrv_{\sag} = \sig \mid \Wosrv = \wos]$ is the probability that an agent receives signal $\sig$ under world state $\wos$. The signal distributions $((P_{hH}, P_{lH}), (P_{hL}, P_{lL}))$ are common knowledge.  The signals are positively correlated to the world states. We assume $P_{hH} > P_{hL}$ and $P_{lH} < P_{lL}$. On the other hand, we allow biased signals and DO NOT assume $P_{hH} > P_{lH}$ or $P_{hL} < P_{lL}$, which is a stronger assumption. 

\begin{figure*}[htp]
    \centering
    \includegraphics[width = 0.99\linewidth]{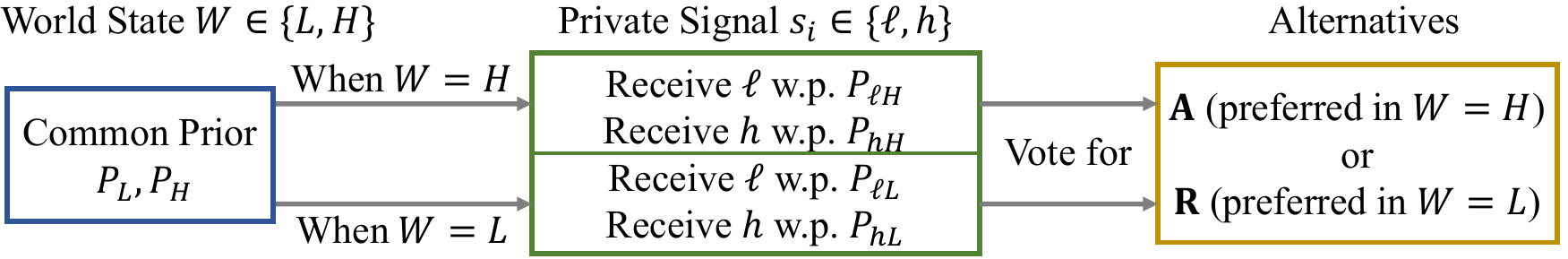}
    \caption{The information structure of the world state and private signals.}
    \label{fig:signal}
\end{figure*}

\paragraph{Utility and Types of Agents.} Each agent $\sag$ has a utility function on the true world state and the outcome of the vote. 
Here we assume agents have 0-1 utility functions to convey the main idea of our work while hiding much of the complexity, 
Formally, we have  $\vt_{\sag}: \Wosset \times \{\bA, \bR\}\to \{0,1\}$. The extension to the setting with general integer utility functions is in Section~\ref{sec:general_util}. 

Agents can be categorized into three types by their preferences. 
{\em Predetermined} agents always prefer the same alternative, and {\em contingent} agents have preferences depending on the world state. Predetermined agents can be further categorized into {\em friendly} and {\em unfriendly} agents based on the alternative they prefer. The utility function of each type of agents is in Table~\ref{tab:util}. 

\noindent\begin{minipage}{\linewidth}
\vspace{0.3cm}
\centering
\begin{tabular}{@{}ccccc@{}}
\toprule
Type         & $\vt_i(H,\bA)$         & $\vt_i(L,\bA)$         & $\vt_i(H,\bR)$         & $\vt_i(L,\bR)$         \\ \midrule
Friendly & 1          & 1           & 0           & 0           \\
Unfriendly & 0           & 0           & 1           & 1          \\
Contingent & 1           & 0           & 0          & 1           \\ \bottomrule
\end{tabular}
\captionof{table}{Utility of agents.\label{tab:util}}
\end{minipage}


Let $F, U, C$ be the set of three types of agents respectively, and $\agf, \agu$, and $\agc$ be the approximated fraction of each type of agent. Formally, given $\ag$ agents, $\tf = \lfloor \agf\cdot \ag \rfloor$ is the number of friendly agents, $\tu = \lfloor \agu\cdot \ag \rfloor$ is the number of unfriendly agents, and $\tc = \ag - \tf - \tu$ is the number of contingent agents.  $\agf, \agu$, and $\agc$ are common knowledge and do not depend on $\ag$.

\paragraph{Informed Majority Decision.}
The goal of the voting is to output the {\em informed majority decision}, the alternative favored by the majority given the world state is known to all the agents. We assume that neither $\agf$ nor $\agu$ is larger than 0.5. Otherwise, friendly or unfriendly agents can dominate the voting outcome, and the alternative favored by the majority will be independent of the world state. Therefore, in the rest of this paper, $\bA$ is the informed majority decision when the world state is $H$, and $\bR$ is the informed majority decision when the world state is $L$. 

\begin{ex}[Information Structure]
    \label{ex:setting} Consider the COVID-19 policy-making scenario. Suppose $\ag = 20$ voters are deciding whether to accept ($\bA$) or reject ($\bR$) the policy.World state $H$ represents a viral pandemic, and world state $L$ represents the opposite. The common prior on the world state is $P_H = 0.6$, and $P_L = 0.4$. 
    The signals are noisy information about the world state. Suppose the community exaggerates the pandemic, and the signal is biased. The distribution is $P_{hH} = 0.8$, $P_{lH} = 0.2$, $P_{hL} = 0.6$, and $P_{lL} = 0.4$. An agent receives an $h$ signal with probability $0.8$ when the world state is $H$ and probability $0.6$ when the state is $L$. 

    Among 20 agents, there are $5$ friendly agents, $6$ unfriendly agents, and $9$ contingent agents. Neither friendly nor unfriendly agents consist of the majority. Therefore, the informed majority decision is $\bA$ when the world state is $H$ and $\bR$ when the world state is $L$. \hfill\qed

\end{ex}

\paragraph{Sequence of Environments.}
An \textinst{} $\instance$ is a set of parameters of a voting game, which contains the agent number $\ag$, the world state prior distribution $(P_L, P_H)$, the signal distributions $(P_{hH}, P_{lH})$ and $(P_{hL}, P_{lL})$,  and the approximated fraction of each type of agents $(\agf, \agu, \agc)$. 

In this paper, we study the problem asymptotically. Therefore, we define a sequence of \textinst{}s in which the number of voters increases while most other parameters remain the same. 
Let $\{\instance_\ag\}_{\ag=1}^{\infty}$ (or $\{\instance_\ag\}$ for short) be a sequence of \textinst{}s, where each $\instance_\ag$ is an \textinst{} of $\ag$ agents.
The \textinst{}s in a sequence share the same prior distributions, signal distributions, and approximated fractions of each type.

\subsection{One round voting}

\paragraph{One-round voting.} In the one-round voting game, every agent casts a vote for either $\bA$ or $\bR$. The alternative getting more than a $0.5$ fraction of votes becomes the winner. 

\paragraph{Strategy in the one-round voting game.} We use $\stgbar$ to denote the strategy of an agent in the one-round voting game. $\stgbar$ is a mapping from an agent's signal to a probability distribution on $\{\bA, \bR\}$ representing her action. 
A strategy is pure if the vote is deterministic. A strategy profile $\stgpbar$ is the vector of strategies of all agents. A strategy profile is pure if it contains only pure strategies and symmetric if all agents with equal utility function play the same strategy. When agents have 0-1 utilities, symmetricity is equivalent to agents of the same type playing the same strategy. 

\begin{definition}[Informative Voting]
    An agent votes informatively if he/she always votes for $\bA$ when his/her signal is $h$ and always votes for $\bR$ when his/her signal is $l$.
\end{definition}

\paragraph{Expected utility} Let $\lp_{\wos}^{\bA}(\stgpbar)$ ($\lp_{\wos}^{\bR}(\stgpbar)$, respectively) be the (ex-ante, before agents receiving their signals) probability that $\bA$ ($\bR$, respectively) becomes the winner in the one round voting when the world state is $\wos$.
The expected utility of an agent $\sag$ exclusively depends on $\lp_{\wos}^{\bA}(\stgpbar)$ and $\lp_{\wos}^{\bR}(\stgpbar)$:
\begin{align*}
    \ut_{\sag}(\stgpbar) = &\ P_{L}(\lp_{L}^{\bA}(\stgpbar)\cdot\vt_{\sag}(L, \bA) + \lp_{L}^{\bR}(\stgpbar)\cdot \vt_{\sag}(L, \bR))\\
    &\ +  P_{H}(\lp_{H}^{\bA}(\stgpbar)\cdot\vt_{\sag}(H, \bA) + \lp_{H}^{\bR}(\stgpbar)\cdot \vt_{\sag}(H, \bR)). 
\end{align*}

\paragraph{Sequence of Strategy Profiles.}
We define a sequence of strategy profiles $\{\stgpbar_\ag\}_{\ag=1}^{\infty}$ on an \textinst{} sequence $\{\instance_\ag\}$, where for each $\ag$, $\stgpbar_{\ag}$ is a strategy profile of a one-round voting game with parameters in $\instance_{\ag}$. We do not have additional assumptions about the agents. Therefore, for different $\ag$, the strategies of agents can be drastically different without further specification. 

\paragraph{$\varepsilon$-strong Bayes Nash Equilibrium (BNE)} 
Due to the nature of the majority rule, deviation from a single agent can unlikely change the outcome. Therefore, we take group strategic behavior of agents in to consideration and adopt the solution concept of (approximated) strong Bayes Nash equilibrium.
In an $\varepsilon$-strong Bayes Nash Equilibrium, no group of agents can increase their utilities by more than $\varepsilon$ through deviation. A strategy profile $\stgpbar = (\stgbar_1, \stgbar_2,\cdots, \stgbar_\ag)$ is an {\em $\varepsilon$-strong Bayes Nash Equilibrium} ($\varepsilon$-strong BNE) if there does not exist a subset of agents $D$ (the {\em deviating group}) and a strategy profile $\stgpbar' = (\stgbar_1', \stgbar_2',\cdots, \stgbar_\ag')$ (the {\em deviating strategy profile}) such that
\begin{enumerate}
    \item $\stgbar_\sag = \stgbar_\sag'$ for all $\sag\not\in D$; 
    \item $\ut_{\sag}(\stgpbar') \ge \ut_{\sag}(\stgpbar)$ for all $\sag\in D$; and 
    \item there exists $\sag\in D$ such that $\ut_{\sag}(\stgpbar') > \ut_{\sag}(\stgpbar) +\varepsilon$. 
\end{enumerate}

By definition, when $\varepsilon=0$, the equilibrium is a strong Bayes Nash Equilibrium where no group of agents can strictly increase their utilities through deviation. 

\section{Equilibria in One-Round Voting}
\label{sec:one_round_unnatural}

\citet{han2023wisdom} show that a strategy profile where predetermined agents play their dominant strategies leads to the informed majority decision with high probability if and only if it is an $\varepsilon$-strong BNE.
This implies the informed majority decision is (approximately) reached for every equilibrium in the one-round voting game. 
However, the equilibria in the one-round voting game can be complicated and unnatural.
In this section, we show scenarios where natural strategy profiles fail to reach high fidelity (and thus are not approximate strong BNEs according to the if-and-only-if characterization) and how complex the existing equilibria are. 

Firstly, in games with only contingent agents  (so all agents have the same object), NO pure symmetric equilibria exist in the one round voting game in ANY environment with biased signals. 


\begin{prop}
    \label{prop:unnatural}
    Suppose $\agc = 1$. Then for any sequence of \textinst{}s $\{\instance_n\}$ such that $P_{hL} > 0.5$ or $P_{lH} > 0.5$, 
    no sequence pure symmetric strategy profiles $\{\stgpbar_\ag\}$ of the one-round voting game satisfies that $\stgpbar_{\ag}$ is a $\varepsilon$-strong BNE for every $\ag$ with $\varepsilon \to 0$.  
\end{prop}

\begin{proof}


Without loss of generality, we consider the case that $P_{hL} > 0.5$, and the reasoning for the case that $P_{lH} > 0.5$ is similar. 
There are in total four pure symmetric strategy profiles: all agents always vote for $\bA$, always vote for $\bR$, vote informatively, and vote against their private signals (vote for $\bA$ when receiving signal $\ell$ and for $\bR$ when receiving signal $h$). Always voting for an alternative has a probability of $P_H$ or $P_L$ to miss the informed majority decision and cannot be an equilibrium. 
When $P_{hL} > 0.5$, agents are more likely to receive $h$ than $\ell$ on world state $L$. Therefore, by applying the Hoeffding inequality, informative voting does not lead to the informed majority decision in the world state $L$ with high probability and cannot be an equilibrium. Similarly, voting against signals cannot be an equilibrium as $P_{hH}> 0.5$.
\end{proof}

More characterizations that natural strategies --- for example, sincere voting ---  fail to be equilibria in one-round voting are in Appendix~\ref{apx:one_round}. 

In the \textinst{}s where pure symmetric strategies fail to be equilibria, reaching the informed majority decision requires agents to artificially ``shift'' the probabilities in the strategy. 
For example, consider a biased setting where the common prior is $P_H = P_L = 0.5$, and the signal distribution is $P_{hH}=0.9, P_{lH}=0.1, P_{hL}=0.7$, and $P_{lL}=0.3$. An equilibrium under this setting is a symmetric mixed profile where agents with signal $l$ always vote for $\bR$, and agents with $h$ vote for $\bA$ with probability 0.6. 
In this strategy profile, the expected fraction of $\bA$ votes in the world $H$ is $0.6\times 0.9=0.54>0.5$ and the expected fraction of $\bA$ votes under world $L$ is $0.6\times 0.7=0.42<0.5$.
Alternatively, instead of playing a mixed strategy, agents can play a pure asymmetric profile where $40\%$ of agents always vote for $\bR$. 

These ``successful'' strategy profiles, however, are often too complex for agents to agree on.
For the mixed strategy profile, even if agents are clever enough to be aware that a probability shift needs to be done, different agents may choose different probability shifts, which may still lead to the undesirable outcome.
For the pure asymmetric profile, without an intrinsic incentive, it is difficult and unclear for identical agents to coordinate who casts the ``opposite vote'' $\bR$ when receiving signal $h$.

}

\section{Two-round Voting and Equilibria}
\label{sec:model}
{
In this section, we introduce our two-round voting mechanism. One-round and two-round voting share the same \textinst{}, including the number of agents, distributions, the fraction of each type, and utilities. Their difference lies in the game form and the strategies. 

\paragraph{Two-round Voting.} Agents play a two-round anonymous voting game where the winner is completely determined by the second round. The five steps are shown in Algorithm~\ref{alg:two-round}. 
\begin{algorithm}[htbp]
\caption{Two-round voting mechansim}\label{alg:two-round}
    \begin{algorithmic}[1]
    \STATE All the agents receive their signals. 
    \STATE \textbf{First round} (Polling stage): agents cast the first round vote for $\bA$ or $\bR$. 
    \STATE Agents observe the number of votes for $\bA$ in the first round denoted by $\hst \in \{0, 1, \cdots, \ag\}$.
    \STATE \textbf{Second Round} (Voting stage): agents cast the second round vote for $\bA$ or $\bR$. 
    \STATE The alternative that gets more than a $0.5$ fraction of votes in the second round becomes the winner.
\end{algorithmic}
\end{algorithm}

\paragraph{Strategy.}
A strategy in a two-round voting game  $\stg_{\sag}$ is a pair $(\stgone_{\sag}, \stgtwo_{\sag})$. The first-round strategy $\stgone_{\sag}$ is a mapping from agent $\sag$'s signal to a distribution on $\{\bA, \bR\}$ denoting her action in the first round. The second-round strategy $\stgtwo_{\sag}$ is a mapping from agent $\sag$'s signal AND the result of the first-round vote $\hst$ to a distribution on $\{\bA, \bR\}$ denoting her action in the second round. 
An agent $\sag$ votes {\em informatively} in the first round if he/she always votes for $\bA$ when his/her signal is $h$ and always votes for $\bR$ when his/her signal is $l$.
A strategy profile $\stgp$ is the vector of strategies of all agents. 
The definition of pure and symmetric strategy profile follows that in the one-round voting. 

\paragraph{Sequence of Strategy Profile}
We define a sequence of two-round voting strategy profiles $\{\stgp_\ag\}_{\ag=1}^{\infty}$ on an \textinst{} sequence $\{\instance_\ag\}$, where for each $\ag$, $\stgp_{\ag}$ is a strategy profile of a two-round voting game with parameters in $\instance_{\ag}$. 


\paragraph{Expected Utility and Fidelity.} 
Given a strategy profile $\stgp$, let $\lp_{\wos}^{\bA}(\stgp)$ ($\lp_{\wos}^{\bR}(\stgp)$, respectively) be the---ex-ante, before agents receiving their signals---probability that $\bA$ ($\bR$, respectively) becomes the winner when the world state is $\wos$. 

The expected utility in the two-round voting is in the same form as that in the one-round voting.
\begin{align*}
    \ut_{\sag}(\stgp) = &\ P_{L}(\lp_{L}^{\bA}(\stgp)\cdot\vt_{\sag}(L, \bA) + \lp_{L}^{\bR}(\stgp)\cdot \vt_{\sag}(L, \bR))\\
    &\ +  P_{H}(\lp_{H}^{\bA}(\stgp)\cdot\vt_{\sag}(H, \bA) + \lp_{H}^{\bR}(\stgp)\cdot \vt_{\sag}(H, \bR)). 
\end{align*}

{\em Fidelity} is the likelihood that the informed majority decision is reached. In the two-round voting game, the fidelity of a strategy profile $\stgp$ is as follows. 
\begin{align*}
    \acc(\stgp) = &\ P_{L}\cdot \lp_{L}^{\bR}(\stgp) + P_{H} \cdot \lp_{H}^{\bA}(\stgp). 
\end{align*}

For contingent agents with 0-1 utility, their expected utility function $\ut(\stgp) = \acc(\stgp)$. 




An $\varepsilon$-strong Bayes Nash Equilibrium in the two-round voting has the same definition as that in the one-round voting except that the components (strategies, expected utilities) are those in the two-round voting game. 



\paragraph{Some Remarks.}
The game of two-round voting is not exactly a standard \emph{Bayesian game}, as the strategy of an agent depends not only on the signal received but also on the outcome of the first round.
Our definitions of strategies, utility functions, and equilibria are straightforward and natural extensions of the standard Bayesian games.

Another natural formulation for the two round-voting is an \emph{extensive-form game}.
We believe this is an interesting direction for future work.
However, standard equilibrium solution concepts such as \emph{subgame perfect Nash equilibrium} and \emph{perfect Bayesian equilibrium} fail to capture the feature of group deviations.
The first step in this future direction is to propose a reasonable equilibrium concept that extends strong Bayes Nash equilibrium to extensive-form games.}
{
\subsection{Equilibria in Two Round Voting}
We show that the two-round voting mechanism is indeed a simple mechanism with a natural equilibrium and leads to the informed majority decision. 
Our first result shows that every $\varepsilon$-strong BNE in the two-round voting mechanisms with $\varepsilon$ converges to 0 has fidelity converging to 1. This guarantees that the mechanism can always achieve an informed majority decision under the strategic behavior of agents. 


\begin{thm}[{\bf All $\varepsilon$-BNE are good}{}]
    \label{thm:eqtoacc}
Let $\{\stgp_{\ag}\}$ be a sequence of profiles in the two-round voting game. If for any $\ag$, $\stgp_{\ag}$ is an $\varepsilon$-strong BNE with $\varepsilon = o(1)$, then $\lim_{\ag\to\infty} \acc(\stgp_\ag) = 1$.
\end{thm}



Moreover, we show the existence of a natural equilibrium in the two-round mechanism. While the classical informative voting or sincere voting may not be an equilibrium in the one-round voting game, their combination becomes an equilibrium as follows. 

\paragraph{``Informative + Sincere'' voting} All the friendly agents always vote for $\bA$ in both rounds, and all the unfriendly agents always vote for $\bR$ in both rounds. All the contingent agents vote informative in the first round and vote sincerely based on their observation of the first round outcome. That is, conditioned on the first-round strategies and the first-round outcome $\hst$, a contingent agent $i$ with private signal $\sig$ updates his/her belief on the world state. Then the agent votes for $\bA$ if and only if he/she believes the world state is more likely to be $H$, i.e., $\Pr[H \mid \sig, \hst] \ge \frac12$. 

\begin{thm}[{\bf``Informative + sincere'' equilibrium}{}]
    \label{thm:sincere}
    For any \textinst{} sequence $\{\instance_{\ag}\}$, the sequence of the ``information + sincere'' two-round voting strategy profiles $\{\stgp^{\dag}_{\ag}\}$ satisfies that (1) $\lim_{\ag\to\infty}\acc(\stgp^{\dag}_{\ag}) = 1$,  and (2) for every $\ag$, $\stgp^{\dag}_{\ag}$ is a $\varepsilon$-strong BNE where $\varepsilon$ converges to 0 as $\ag$ goes to $\infty$. 
\end{thm}

The rationale behind the ``informative + sincere'' strategy is straightforward---by observing the outcome of the first-round informative voting, every contingent agent knows the private information of all contingent agents anonymously, which informs them of the correct world state with high probability. Therefore, in the second round, they are nearly certain of their preferred alternative and vote accordingly.

While agents may not necessarily perform Bayesian updates in the votes, we show that the intuition of ''sharing information and getting updated'' derives a large class of ``informative + threshold'' equilibria. Therefore, as agents still follow this ``informative + threshold'' pattern, the informed majority can be reached. 

In an ``informative + threshold'' strategy profile, all the predetermined agents vote for their preferred alternative in both rounds; every contingent agent $i$ votes informatively in the first round and votes for $\bA$ in the second round if and only if the first round outcome $\hst$ exceeds some threshold $\hst^*_i$. For the ``informative + sincere'' voting, the threshold $\hst^*$ is determined by the constraint $\Pr[H \mid \sig, \hst] \ge \frac12$ as agents apply the Bayes Theorem to update their belief. 
The explicit form is
\begin{equation*}
    \hst^* = \tf + \frac{\log \frac{P_L}{P_H} + \tc \cdot \log \frac{P_{\ell L}}{P_{\ell H}}}{\log \frac{P_{hH}P_{\ell L}}{P_{\ell H}P_{hL}}}. 
\end{equation*}

\begin{ex}
    We calculate the ``informative + sincere'' strategy under the setting in Example~\ref{ex:setting}. In this \textinst{}, $\tf = 5$, $\tc = 9$, $P_H = 0.6$ $P_{hH} = 0.8$, and $P_{hL} = 0.6$. Therefore, $\hst^* = \frac{\log \frac23 + 9\times \log 2}{\log \frac83} \approx 10.95$. Therefore, the contingent voters who vote sincerely in the second round will vote for $\bA$ if and only if there are at least 11 votes for $\bA$ in the first round. \hfill\qed
\end{ex}

Additionally, ``informative + threshold'' class also includes more ``natural'' strategy profiles, such as the ``informative + surprisingly popular'' strategy profile.
\begin{ex}
    The ``informative + surprisingly popular'' strategy shares the similar idea with the ``surprisingly popular'' mechanism~\citep{prelec2017solution}. The threshold of every agent $i$ is the prior expectation on the first-round outcome $\hst$, which is $\hst^* = \tf + \tc \cdot (P_H \cdot P_{hH} + P_L\cdot P_{hL})$. This comes from that the $\tf$ friendly agents always vote for $\bA$, and the expected share of $\tc$ contingent agents voting for $\bA$ is $P_{hH}$ under world state $H$ and $P_{hL}$ under world state $L$.When the number of votes for $\bA$ exceeds the expectation, $\bA$ becomes the ``surprisingly popular'' alternative in the first round and is elected the winner in the second round. Otherwise, $\bR$ becomes the ``surprisingly popular'' alternative. \hfill\qed
\end{ex}

The ``informative + threshold'' strategy profiles share the merit the ``informative + sincere'' equilibrium has. Every agent votes deterministically. When agents with equal utility functions adopt the same threshold, they are also symmetric. Most importantly, every ``informative + threshold'' strategy profile in which all thresholds lie between the expected first-round outcome conditioned on world state $H$ and on world state $L$ is an equilibrium that leads to the informed majority decision. 

\begin{dfn}
    We say an sequence of ``informative + threshold'' strategies is constantly separable if there exists a $\ag_0 > 0$ and a constant $\delta  > 0$ such that for every $\ag \ge \ag_0$ and every agent $i$, $\tf + \tc\cdot P_{hL} + \delta \cdot \ag\le \hst_i^* \le \tf + \tc \cdot P_{hH} - \delta\cdot \ag$.
\end{dfn}
As illustrated in Figure~\ref{fig:naturaleq}, the expected votes for $\bA$ in the first round is $\tf + \tc \cdot P_{hH}$ under world state $H$ and $\tf + \tc \cdot P_{hL}$ under world state $L$. 

\begin{thm}[{\bf``Informative + threshold'' equilibrium}{}]
    \label{thm:threshold}
    Under any \textinst{} sequence $\{\instance_{\ag}\}$, every constantly separable sequence of profiles $\{\stgp^*\}$ satisfies that (1) $\lim_{\ag\to\infty}\acc(\stgp^{*}_{\ag}) = 1$,  and (2) for every $\ag$, $\stgp^{*}_{\ag}$ is a $\varepsilon$-strong BNE where $\varepsilon$ converges to 0 as $\ag$ goes to $\infty$. 
\end{thm}

\begin{figure*}[htbp]
    \centering
    \includegraphics[width = 0.99
    \linewidth]{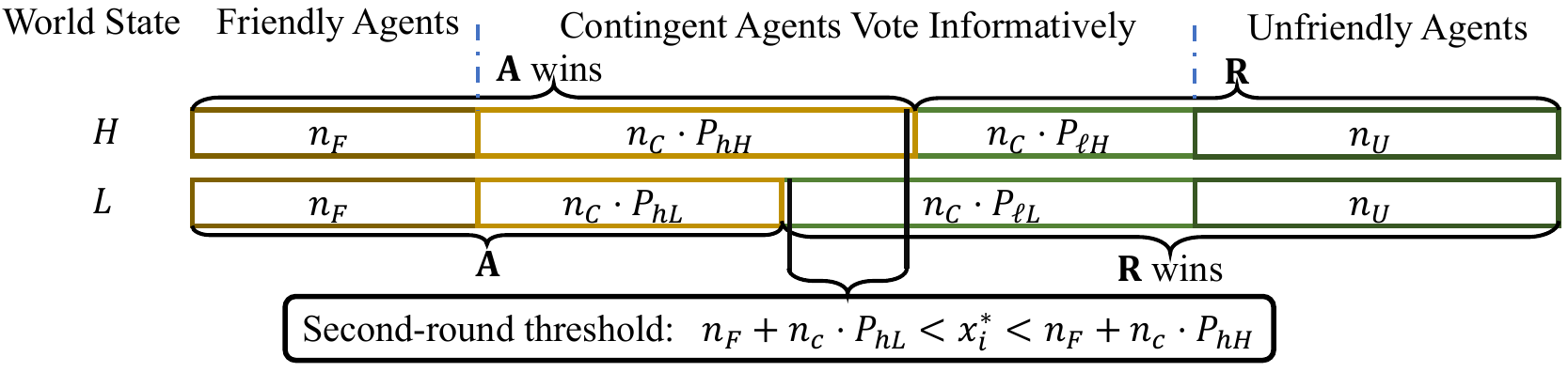}
    \caption{Illustration of the range of the second round threshold for an ``informative + threshold'' equilibrium. }
    \label{fig:naturaleq}
\end{figure*}

\paragraph{Remark on Theorem~\ref{thm:eqtoacc}.} 
Theorem~\ref{thm:eqtoacc} states that any $o(1)$-strong BNE leads to nearly optimal fidelity. Unlike the one-round voting, the other direction---optimal fidelity implying equilibria---does not hold in the two-round voting game, even when predetermined agents only play their dominant second-round strategies. 
Intuitively, some supportive behaviors of the predetermined agents in the first round can also constitute a strategy profile that leads to optimal fidelity, while predetermined agents have incentives to deviate from such behaviors to more self-interested strategies. 

For example, friendly and unfriendly agents voting informatively in the first round can help contingent agents infer the world state and lead to an informed majority decision. However, this does not form an equilibrium, as friendly agents have incentives to always vote for $\bA$ in the first round and convince contingent agents that the world state is $H$. Similarly, unfriendly agents have incentives to vote for $\bR$.

\subsection{Proof Sketch for Theorem~\ref{thm:eqtoacc}}
Suppose $\{\stgp_{\ag}\}$ is a sequence of strategy profiles whose fidelity does not converge to 1. We show that there exists a different sequence of profiles $\{\stgp_{\ag}'\}$ such that agents have incentives to deviate from $\stgp_\ag$ to $\stgp'_\ag$ for infinitely many $\ag$. 

We define a strategy profile in the two-round voting as {\em regular} if all friendly agents always vote for $\bA$ and all unfriendly agents always vote for $\bR$ in the {\em second} round. 
We first give the proof when $\{\stgp_{\ag}\}$ is a regular profile sequence, which characterizes the main idea of the proof while reducing complication, and then extend the proof to non-regular profiles. The full proof can be found in Appendix~\ref{apx:eqtoacc}. 

\subsubsection{Proof for regular profiles.}
We take an infinite set $\mathcal{N} \subseteq \mathbb{N}_+$ and a constant $\delta > 0$ such that for all $\ag \in \mathcal{N}$, the fidelity of $\stgp_{\ag}$ does not exceed $1 - \delta$. Let the deviating group $D$ be all the contingent agents, who prefer $\bA$ in world state $H$ and $\bR$ in world state $L$. Therefore, if the contingent agents switching from $\stgp_{\ag}$ to $\stgp_{\ag}'$ increases the fidelity, they get higher expected utility and are incentivized to do so. 

In a regular $\stgp_{\ag}$, all the predetermined agents vote for their preferred alternative in the second round. If all the contingent agents vote for the same alternative in the second round, this alternative will be the winner. Therefore, our construction of $\stgp_{\ag}'$ aims to achieve the result that contingent agents reveal the world state with high probability in the first round and vote for the preferred agents in the second round. 

Contingent agents can (collectively, with high probability) distinguish the world states because signal distributions are different in different states of the world. For all the contingent agents, the expected number of $h$ signals is $P_{hH}\cdot \tc$ in world state $H$ and $P_{hL}\cdot \tc$ in world state $L$. If all the contingent agents vote informatively in the first round, then there will be likely more votes for $\bA$ in world state $H$ than in world state $L$ and more $\bR$ votes in world state $L$ than in world state $H$. By the Hoeffding inequality, the number of votes will concentrate to the expectation with high probability as $\ag$ increases. Therefore, contingent agents can play a ``threshold'' strategy in the second round: vote for $\bA$ together if the first round result $\hst$ exceeds the threshold and vote for $\bR$ together otherwise. The threshold should be set between $P_{hH}\cdot \tc$ and $P_{hL}\cdot \tc$. 

However, this strategy still must overcome one difficulty: strategies of predetermined agents in the first round are not guaranteed to be ``regular'' and may offset the difference in the votes of contingent agents. For example, predetermined agents can vote for $\bR$ more often in world state $H$ than in world state $L$ in $\stgp_{\ag}$ so that $\hst$ has a smaller expectation in world state $H$ than in world state $L$. However, contingent agents can ``reverse'' their votes and observations accordingly: they vote opposite to their signal in the first round and vote for $\bA$ if $\hst$ is below a threshold. In this way, they can still distinguish the different world states. 

Now we explicitly present the strategy profile $\stgp_\ag'$. The deviators $D$ are exactly all the contingent agents. Let $\exs_H(\stgp_\ag')$ and $\exs_L(\stgp_\ag')$ be the expectation of the first-round outcome $\hst$ of $\stgp_\ag'$ conditioned on the world state $H$ and $L$ respectively. 

\begin{itemize}
        \item All predetermined agents (they do not deviate) play the same strategy as in $\stgp_{\ag}$. This includes that they play their dominant second-round strategies. 
        \item If more votes for $\bA$ from predetermined agents are expected in world state $H$ than in $L$ in the first round, all contingent agents vote informatively in the first round. For the second round, they vote for $\bA$ if $\hst \ge \frac{\exs_H(\stgp_\ag') + \exs_L(\stgp_\ag')}{2}$ and vote for $\bR$ otherwise.
        \item Otherwise, contingent agents ``reverse'' their strategy. In the first round, every contingent agent votes for $\bA$ if receiving $l$ and for $\bR$ if receiving $h$. In the second round, they vote for $\bA$ if $\hst < \frac{\exs_H(\stgp_\ag') + \exs_L(\stgp_\ag')}{2}$ and $\bR$ otherwise.
    \end{itemize}

As shown in the previous analysis, under $\stgp'_\ag$, contingent agents can distinguish the world state with high probability via the first-round vote, guaranteed by the Hoeffding Inequality. Therefore, the fidelity $\acc(\stgp'_\ag)$ converges to 1 as $\ag$ goes to infinity. 
In this way, for all sufficiently large $\ag$ such that $\acc(\stgp_{\ag})$ is not close to 1, by deviating from $\stgp_{\ag}$ to $\stgp_{\ag}'$, the expected utilities of contingent agents increase by a constant. Therefore, contingent agents have incentives to deviate, and $\stgp_{\ag}$ is NOT an $\varepsilon$-strong BNE for some constant $\varepsilon$.

\subsubsection{Proof Sketch for non-regular sequences.} We adopt the same idea to construct a strategy profile where the deviating group reveals the world state by the first-round vote and votes for the preferred agent in the second round. Then we show that as the fidelity converges to 1, the group has incentives to do so. 

However, for non-regular profiles, the deviating group containing only contingent agents may not work, as there is no assumption on the behavior of predetermined agents in the second round. (For example, if $\agc < 0.5$, and all the predetermined agents always vote for $\bA$, then the winner will always be $\bA$ regardless of the votes from contingent agents.) Fortunately, the following lemma shows that there will always be a type of predetermined agents have incentives to deviate together to reach the informed majority decision. 

\begin{lem} 
\label{lem:positive} 
Let $\{\stgp_{\ag}\}$ be a sequence of profiles such that there exists a constant $\delta$ and an infinite set $\mathcal{N}$ such that for all $\ag\in \mathcal{N}$, $\acc(\stgp_{\ag}) \le 1 - \delta$. 
Then there exists an infinite set $\mathcal{N}' \subseteq \mathcal{N}$ such that for each $\ag \in \mathcal{N}'$, at least one of the following holds for all sequences of profiles $\{\stgp_{\ag}'\}$ with $\lim_{\ag\to\infty} \acc(\stgp'_\ag) = 1$:
\begin{enumerate} 
\item For every friendly agent $\sag_1$, $\ut_{\sag_1} (\stgp_{\ag}') - \ut_{\sag_1} (\stgp_{\ag}) \ge 0$. 
\item For every unfriendly agent $\sag_2$, $\ut_{\sag_2} (\stgp_{\ag}') - \ut_{\sag_2} (\stgp_{\ag}) \ge 0$. \end{enumerate} 
Moreover, for a fixed $\ag$ and $\stgp_{\ag}$, either 1 holds for all $\stgp_{\ag}'$ or 2 holds for all $\stgp'$. 
\end{lem} 

\paragraph{Proof Sketch for Lemma~\ref{lem:positive}. }The fidelity $\acc(\stgp_{\ag})$ is a weighted sum of the likelihood that $\bA$ wins in world state $H$ (denoted as $\lambda_H^\bA(\stgp_\ag)$) and the likelihood that $\bR$ wins in world state $L$ (denoted as $\lambda_L^\bR(\stgp_\ag)$). From $\stgp_\ag$ to $\stgp_\ag'$, both likelihoods increase. If $\lambda_H^\bA(\stgp_\ag)$ increases more, the friendly agents get utility increase; and if $\lambda_L^\bR(\stgp_\ag)$ increases more, the unfriendly agents get utility increase. Moreover, as $\acc(\stgp'_\ag)$ is close to 1, which likelihood increases more solely depends on $\stgp_\ag$. Therefore, which type of agents gets utility increase does not depend on the deviating strategy $\stgp_\ag'$. 


In this way, we consider two deviating strategies. In the first strategy $\stgp_{F}$, all the friendly and contingent agents vote informatively in the first round and play a threshold strategy in the second round. In the second strategy $\stgp_{U}$, all the unfriendly agents and all the contingent agents vote informatively in the first round and play a threshold strategy in the second round. Since $\agf < 0.5$ and $\agu < 0.5$, if contingent agents and one type of predetermined agents vote for the same alternative, that alternative will be the winner. Therefore, with similar reasoning to the regular case, both $\stgp_{F}$ and $\stgp_{U}$ have fidelity converging to 1. Applying Lemma~\ref{lem:positive}, either all friendly agents or all unfriendly agents get the expected utility increase via deviation. Moreover, the type of agents with utility gain is the same in $\stgp_{F}$ and $\stgp_{U}$. Contingent agents will also get a constant increase in expected utility due to the increase in fidelity. Therefore, for each $\stgp_{\ag}$ whose fidelity is not close to 1, agents have incentives to deviate to either $\stgp_{F}$ or $\stgp_{U}$. Therefore, $\stgp_{\ag}$ is NOT an $\varepsilon$-strong BNE with constant $\varepsilon$. 

\subsection{Proof Sketch for Theorem~\ref{thm:sincere} and Theorem~\ref{thm:threshold}.}

As the ``informative + sincere'' strategy profile is also an ``informative + threshold'' strategy profile, it suffices to display the proof for Theorem~\ref{thm:threshold} and shows that the sincere threshold satisfies the constantly separable constraint. 

The proof of Theorem~\ref{thm:threshold} consists of two parts: Lemma~\ref{lem:fidelity_threshold} shows that the fidelity converges to 1, and Lemma~\ref{lem:eq_threshold} shows that the profiles are $o(1)$-strong BNE. The full proof is in Appendix~\ref{apx:threshold}. 


\begin{lem}
    \label{lem:fidelity_threshold}
    For any sequence of instance and any sequence of ``informative + threshold'' strategy profiles $\stgp$, if there exists a constant $\delta > 0$ such that for all $\ag$ and all agent $i$, $\tf + \tc\cdot P_{hL} + \delta \cdot \ag\le \hst_i^\ag \le \tf + \tc \cdot P_{hH} - \delta\cdot \ag$, then $\acc(\stgp)$ converges to 1. 
\end{lem} 

Lemma~\ref{lem:fidelity_threshold} follows similar reasoning in Theorem~\ref{thm:eqtoacc} of showing the deviating strategy profile have fidelity converging to 1. Firstly, the threshold $\hst^*$ is guaranteed to be between the expectation of the first-round vote conditioned on the world state being $H$ and being $L$. By the Hoeffding inequality, the result of the first round concentrates on the expectation with high probability. Therefore, with probability converging to 1, the first-round outcome is larger than all the thresholds under world state $H$ and lower than all the thresholds under world state $L$. This informs all the contingent agents of the world state almost surely, and they will vote for the informed majority decision with high probability.



\begin{lem}
    \label{lem:eq_threshold}
    $\stgp^{*}_{\ag}$ is an $\varepsilon$-strong BNE with $\varepsilon = o(1)$. 
\end{lem}

The proof of Lemma~\ref{lem:eq_threshold} follows a similar idea to the proof by~\citet[Theorem 3.3]{schoenebeck21wisdom}. For arbitrary $\ag$, and a strategy profile $\stgp_\ag$, let $\varepsilon = 2B(B+2)(1 - \acc(\stgp_{\ag}))$, where $B$ is the upper bound of the utility. Then, for any deviating strategy profile $\stgp_{\ag}'$, we show that the deviation will not succeed.
Therefore, $\stgp_{\ag}$ is an $\varepsilon$-strong BNE. Since $\acc(\stgp_{\ag}^*)$ converges to 1, $\varepsilon$ converges to 0. 

\begin{lem}[Informal]
    \label{lem:deviate} For any strategy profile $\stgp$, a deviating group $D$ that have incentives to deviate from $\stgp$ to another strategy profile $\stgp'$, under the approximation parameter $\varepsilon = 6(1 - \acc(\stgp))$ for $\varepsilon$-strong BNE, contains either only friendly agents or only unfriendly agents.
\end{lem}

The proof of Lemma~\ref{lem:deviate} consists of two cases. If $1 - \acc(\stgp') < 2\cdot (1 - \acc(\stgp))$, the fidelities of both profiles are close to 1, and the expected utilities of agents will have no significant difference in two profiles. Therefore, no agents can get more than $\varepsilon$ expected utility gain by deviation. Otherwise, we show that the deviating group contains only one type of predetermined agents. Firstly, contingent agents, whose expected utilities exactly equals to the fidelity, are excluded from the deviating group $D$, as their expected utilities decrease with the fidelity. Then, friendly agents and unfriendly agents cannot be in $D$ simultaneously, as an expected utility gain of more than $\varepsilon$ on one side implies an expected utility decrease on the other side. The formal statement of Lemma~\ref{lem:deviate} is in Appendix~\ref{apx:deviate}. 


Finally, we show that  $D$ contains only friendly agents or only unfriendly agents cannot succeed. 
Without loss of generality, suppose $D$ contains only friendly agents.  If friendly agents deviate in the first round, the result of the first-round voting $x$ will be less likely to reach the threshold, and contingent agents are more likely to vote for $\bR$. If friendly agents deviate in the second round, there will be strictly fewer agents voting for $\bA$ in expectation, and $\bA$ will be less likely to be the winner.
Therefore, any deviation cannot succeed, and $\stgp^{*}_{\ag}$ is a $\varepsilon$-strong BNE with $\varepsilon$ converges to 0. 

Finally, we show that the ``informative+sincere'' threshold is between the two expectation.
\begin{lem}
\label{lem:sincere_is_threshold}
    The sequence of ``informative + sincere'' strategy profiles is constantly separable. 
\end{lem}
Lemma~\ref{lem:sincere_is_threshold} shares the same idea with Lemma~\ref{lem:fidelity_threshold}. Recall that threshold $\hst^*$ is the boundary for $\Pr[H \mid \hst, m] \ge \frac12$. When $\hst$ is close to the upper bound (expected utility when the world state is $H$), $\Pr[H \mid \hst, m]$ is close to 1; and when $\hst$ is close to the lower bound (expected utility when the world state is $L$), $\Pr[H \mid \hst, m]$ is close to 0. Therefore, the threshold must be in between.

\section{Two-Round Voting is More Powerful}
\label{sec:natural}
Section~\ref{sec:one_round_unnatural} and Section~\ref{sec:model} form a sharp contrast between one-round and two-round voting equilibria. In the one-round voting, any \textinst{} with a biased signal structure has no simple and natural strategies form equilibria (and, equivalently, reach the informed majority decision). In the two-round voting, Theorem~\ref{thm:sincere} guarantees a natural strategy profile to reach the informed majority decision and be an equilibrium. 

In this section, we show that in any sequence of \textinst{}s where there exists a one-round equilibrium (where predetermined agents play their dominant strategy), there always exists a two-round equilibrium that will not be more ``complicated'' than the one-round equilibrium. 
Precisely, in this two-round equilibrium, agents disregard the first-round voting and play exactly as the one-round equilibrium in the second round. 
A one-round voting strategy profile is {\em regular} if all predetermined agents play their dominant strategies.

\begin{thm}
\label{thm:one_and_two}
    (Informal) Under the same sequence of \textinst{}s, for any regular one-round voting strategy profile sequence being $o(1)$-strong BNE, any two-round voting profile sequence, where agents play exactly this one-round strategy in the second round, is also an $o(1)$-strong BNE. 
\end{thm}


By not exploiting the information revelation in the first round, the two-round voting game was reduced to the one-round voting game. 
Therefore, we claim that two-round voting is a more powerful tool than one-round voting, as it can reach an informed majority decision whenever one-round voting can reach it, and it provides natural equilibrium in every \textinst{} including those where one-round voting fails to have a natural solution. 

\paragraph{Proof Sketch} The proof proceeds in three steps. Firstly, the fidelity of the two-round profile sequence converges to 1, because agents in the second round behave exactly as they are playing a one-round game, and the fidelity of the one-round profile converges to 1. Secondly, by applying Lemma~\ref{lem:deviate}, a deviating group either contains only friendly agents or only unfriendly agents. Finally, such a deviating group cannot succeed, similar to Lemma~\ref{lem:eq_threshold}. A more technical description of Theorem~\ref{thm:one_and_two} and the full proof is in Appendix~\ref{apx:one_and_two}.

\section{General Valuation Utility.}
\label{sec:general_util}
In this section, we extend our theoretical results to a setting with more general utility functions. 

\paragraph{General Valuation utility} The utility function of each agent $\vt_{\sag}: \Wosset \times \{\bA, \bR\}\to \{0,1, 2, \cdots ,B\}$, where $B > 0$ is the integer upper bound. 
In our model, $\bA$ is more welcomed in $H$, and $\bR$ is more welcomed in $L$: $\vt_{\sag}(H, \bA) \ge \vt_{\sag}(L, \bA)$, and  $\vt_{\sag}(L, \bR) \ge \vt_{\sag}(H, \bR)$. For an agent $\sag$, if $\sag$ is a friendly agent, $\vt_{\sag}(H,\bA) \ge \vt_{\sag}(L, \bA) > \vt_{\sag}(L, \bR) \ge \vt_{\sag}(H, \bR)$; if $\sag$ is an unfriendly agent, $\vt_{\sag}(L, \bR) \ge \vt_{\sag}(H, \bR) > \vt_{\sag}(H,\bA) \ge \vt_{\sag}(L, \bA)$; and if $\sag$ is a contingent agent, $\vt_{\sag}(H, \bA) > \vt_{\sag}(H, \bR)$ and $\vt_{\sag}(L,\bR) > \vt_{\sag}(L, \bA)$. 

Under the general valuation utility setting, agents of the same type have the same ordinal preference, yet their utility functions may be different. For example, the utility function of a contingent agent may not equals to (or be proportional to) the fidelity anymore. All other definitions in the setting remain the same. Specifically, in a symmetric strategy profile, agents with equal utility function play the same strategy. This allows agents with the same ordinal preferences to play differently. 

As the utility diverges, a sincere voter does not always vote for the alternative that is more likely to be the informed majority decision. Instead, they vote for the alternative that brings them a higher expected utility. 
\begin{dfn}[Sincere Voting in the second round.]
    A sincere agent votes as if he/she is making an individual decision. The agent compares his/her expected utility conditioned on his/her private signal $m$, the first-round outcome $\hst$, and the event that $\bA$/$\bR$ becomes the winner, respectively.
    \begin{align*}
    \ut_{\sag}(\bA\mid \hst, \sig) = & \Pr[\Wosrv = L \mid \hst, \sig] \cdot \vt_{\sag}(L, \bA) + \Pr[\Wosrv = H \mid \hst, \sig] \cdot \vt_{\sag}(H, \bA). \\
    \ut_{\sag}(\bR\mid \hst, \sig) = & \Pr[\Wosrv = L \mid\hst, \sig] \cdot \vt_{\sag}(L,\bR) + \Pr[\Wosrv = H \mid \hst, \sig] \cdot \vt_{\sag}(H,\bR). 
\end{align*}

If $\ut_{\sag}(\bA\mid \hst, \sig) \ge \ut_{\sag}(\bR\mid \hst, \sig)$, the agent vote for $\bA$; otherwise, the agent vote for $\bR$. 
\end{dfn}

Now we extend our results to the general valuation setting. The sharp contrast between the one-round voting and the two-round voting still holds. For one-round voting, instances with NO pure symmetric equilibrium still exist. 

\begin{propnb}{prop:unnatural}
\label{prop:unnatural_nb}
There exists a sequence of \textinst{}s $\{\instance_n\}$ such that no sequence pure symmetric strategy profiles $\{\stgpbar_\ag\}$ of the one-round voting game satisfies that $\stgpbar_{\ag}$ is a $\varepsilon$-strong BNE for every $\ag$ with $\varepsilon \to 0$.  
\end{propnb}

The \textinst{}s specified in Proposition~\ref{prop:unnatural}, where all agents are contingent (and have equal utility functions) and the signal is biased, serves as the proof for Proposition~\ref{prop:unnatural_nb}. 

In the sharp contrast, all our positive results in the two-round voting holds in the general valuation utility setting. 

\begin{thmnb}{thm:eqtoacc}[{\bf All $\varepsilon$-BNE are good}{}]
    \label{thm:eqtoacc_nb}
Let $\{\stgp_{\ag}\}$ be a sequence of profiles in the two-round voting game. If for any $\ag$, $\stgp_{\ag}$ is an $\varepsilon$-strong BNE with $\varepsilon = o(1)$, then $\lim_{\ag\to\infty} \acc(\stgp_\ag) = 1$.
\end{thmnb}

\begin{thmnb}{thm:sincere}[{\bf``Informative + sincere'' equilibrium}{}]
    \label{thm:sincere_nb}
    For any \textinst{} sequence $\{\instance_{\ag}\}$, the sequence of the ``information + sincere'' two-round voting strategy profiles $\{\stgp^{\dag}_{\ag}\}$ satisfies that (1) $\lim_{\ag\to\infty}\acc(\stgp^{\dag}_{\ag}) = 1$,  and (2) for every $\ag$, $\stgp^{\dag}_{\ag}$ is a $\varepsilon$-strong BNE where $\varepsilon$ converges to 0 as $\ag$ goes to $\infty$. 
\end{thmnb}

\begin{thmnb}{thm:threshold}[{\bf``Informative + threshold'' equilibrium}{}]
\label{thm:threshold_nb}
    Under any \textinst{} sequence $\{\instance_{\ag}\}$, every constantly separable sequence of profiles $\{\stgp^*\}$ satisfies that (1) $\lim_{\ag\to\infty}\acc(\stgp^{*}_{\ag}) = 1$,  and (2) for every $\ag$, $\stgp^{*}_{\ag}$ is a $\varepsilon$-strong BNE where $\varepsilon$ converges to 0 as $\ag$ goes to $\infty$. 
\end{thmnb}

\begin{thmnb}{thm:one_and_two}
    (Informal) Under the same sequence of \textinst{}s, for any regular one-round voting strategy profile sequence being $o(1)$-strong BNE, any two-round voting profile sequence, where agents play exactly this one-round strategy in the second round, is also an $o(1)$-strong BNE. 
\end{thmnb}

\paragraph{Remark for Theorem~\ref{thm:sincere_nb}} Under the general valuation setting, as the behavior of sincere agents diverges with the utility, the ``sincere'' threshold depends on the utility of the agents. Formally, for a contingent agent $i$, 
\begin{equation*}
    \hst ^*_i = \tf + \frac{\log\frac{P_L \cdot (\vt_i(L, \bR) - \vt_i(L, \bA))}{P_H \cdot (\vt_i(H, \bA) - \vt_i(H, \bR))} + \tc \cdot \log \frac{P_{\ell L}}{P_{\ell H}}}{\log \frac{P_{hH}P_{\ell L}}{P_{\ell H}P_{hL}}}.
\end{equation*}
Agents with different utility functions may have different threshold $\hst^*_i$, yet this gap remains constant as $\ag$ grows. Therefore, in large elections, sincere voters adopt different yet close thresholds. 

The proofs of Theorem~\ref{thm:eqtoacc_nb} to~\ref{thm:one_and_two}' follow a similar reasoning to their 0-1 utility version, yet the general valuation utilities bring new technical challenges. As the preferences vary, the utility change of the agents are less aligned with the change in the fidelity. 
For example, in Lemma~\ref{lem:positive}, we show that a type of predetermined agents are happy with fidelity increases. However, in the general valuation setting, the Lemma further relies on the assumptions that agents have different sensitivities of utilities under different world state. The friendly agents are more sensitive in world state $H$ ($\vt_{\sag}(H,\bA) - \vt_{\sag}(H, \bR) > \vt_{\sag}(L, \bA) - \vt_{\sag}(L, \bR)$), while unfriendly agents are more sensitive in world state $L$. Therefore, they usually have higher utility increase when the fidelity increases compared to the utility decrease when the fidelity decreases, which improve their willingness to deviate to increase expected utility. All the proofs are in Appendix~\ref{apx:proof}

}

\section{Experiments}
{

In addition to our theoretical results, we conducted simulated voting experiments with generative AI agents to compare one-round and two-round voting. Specifically, we inject generative AI bots with different preferences and information drawn from given distributions, then run both one-round and two-round voting under various \textinst{} parameters.
Our findings show that the generative AI agents achieved the informed majority decision more often under the two-round mechanism than under the one-round mechanism. 


\subsection{Experimental Setting}

Our experiments employed two generative AI models—OpenAI GPT-4o and Deepseek R1, a newly released reasoning model. In each voting scenario, all voters were drawn from the same model.

For each vote, we set $\ag = 40$ agents, including $5$ friendly agents, $3$ unfriendly agents, and $32$ contingent agents. The prior is fixed at $P_H = 0.5$. We run experiments on two different signal distributions, the unbiased distribution $P_{hH} = 0.8, P_{hL} = 0.2$ and the biased distribution $P_{hH} = 0.8, P_{hL} = 0.6$.
We also test a wider range of signal distributions with fewer samples and 
select the distributions that capture different scenarios anticipated by our theoretical findings on the performance of one-round versus two-round voting. Under the unbiased distribution, both mechanisms 
admit natural equilibria that achieve the informed majority decision, so we expect strong performance from both. In contrast, the biased distribution, as shown by Proposition~\ref{prop:unnatural}, admits no pure-symmetric equilibria for one-round voting, whereas ``informative + sincere'' remains viable in the two-round mechanism—implying better performance in that setting.

We run 50 votes for each set of parameters.
In each vote, we first sampled the world state and signals for every agent from the designated distribution. Using these same samples, we then ran both a one-round and a two-round voting procedure. Specifically, we create $\ag$ generative AI agents for the one-round voting and another 
$\ag$ agents for the two-round voting.
In one-round voting, each agent was asked only once for its vote. In two-round voting, each agent voted twice: we first collected every agent’s first-round vote simultaneously, then shared the first-round outcome with all agents (including their own first-round response as a history message) when asking for the second-round vote.


The prompt for AI agents contains the following components. 

    \textit{Voting scenario.} We chose a company hiring setting: each AI agent plays the role of a human resource specialist and votes on whether to hire a candidate.
    
    \textit{World state and signals.} The world state is whether the candidate is qualified, which is not revealed to the agents. Each AI agent is told to have a Good or Bad impression of the candidate, representing its private signal.
    
    \textit{Prior and signal distribution.} AI agents are informed of the prior and signal distribution in two different ways. In the {\em accurate distribution} setting, every agent is told the exact distribution $(P_H, P_L)$ and $(P_{hH}, P_{hL}, P_{\ell H}, P_{\ell L})$. In the {\em vague distribution} setting, the information is transferred into a vague description via a Deepseek R1 chatbot. For example, the distribution $P_H = 0.5$, $P_{hH} = 0.8$, and $P_{hL} = 0.6$ is converted the following description. 

    \begin{ex}[Vague information in the prompt.]
Before meeting with the candidate, you thought she had a balanced chance to be qualified.
If the candidate is qualified, you will likely have a good impression after meeting with her with a high probability well above half.
If the candidate is not qualified, you will be less likely, though still moderately probable, to have a good impression at a likelihood distinctly lower than the qualified scenario but remaining above the halfway mark.\hfil\qed
\end{ex}
    
    \textit{Types and preferences.} Every agent is told their preferences based on their type. 
    
   \textit{Voting Mechanisms.} Every agent is informed that it will participate in either a one-round or a two-round voting procedure. For the two-round voting, we clarify that the first round does not determine the winner but is publicly observable.
   
    \textit{Vote.} In both the one-round voting and each round of the two-round voting, every agent is asked to determine and report its best vote (YES or NO), along with a brief explanation of its reasoning.


\subsection{Experimental Results}

\begin{table}[htbp]
\begin{tabular}{@{}cccccc@{}}
\toprule
\multirow{2}{*}{Model} & Distribution & \multicolumn{2}{c}{Unbiased (0.8, 0.2)} & \multicolumn{2}{c}{Biased  (0.8, 0.6)} \\ \cmidrule(l){2-6} 
 & Information & Accurate & Vague & Accurate & Vague \\ \midrule
\multirow{2}{*}{Deepseek R1} & One-round & 1 & 1 & \textbf{0.56} & \textbf{0.52} \\
 & Two-round & 1 & 1 & \textbf{0.88} & \textbf{0.78} \\
\multirow{2}{*}{GPT-4o} & One-round & 1 & 1 & 0.5 & 0.52 \\
 & Two-round & 1 & 1 & 0.48 & 0.5 \\ \bottomrule
\end{tabular}
\caption{The likelihood of the informed majority decision is reached. \label{tbl:GPT_IMD}}
\end{table}


In Table~\ref{tbl:GPT_IMD}, we compare how often one-round versus two-round voting achieves the informed majority decision under different signal distributions, information conditions, and generative AI models. When signals are unbiased, both voting mechanisms always reach the informed majority decision. However, their performances diverge significantly under biased signals. In particular, Deepseek R1 agents show a marked accuracy gap: one-round voting frequently misses the informed majority decision when the world state is unqualified (even though agents still receive a "Good" signal with probability 0.6), whereas two-round voting retains good performance under both world states. On the other hand, for the non-reasoning model GPT-4o, there is no significant different between the performance of one-round voting and two-round voting. 



In addition to the analysis of the informed majority decision, we also examine the behaviors of the AI agents from their reasoning and analysis and have some interesting observations. 

\paragraph{Bayesian Updates with accurate and vague information.} Almost every agent performs a Bayesian update to infer the likelihood of each world state based on the signal distribution and the private information it obtains. Even when vague information is given, many agents tend to guess a concrete value for each description and perform a Bayesian update based on the guessed values. 

universal switch in the second-round votes where a voter casts a different vote to its first-round vote. Examining the analysis of these voters, they usually find considerable votes in the first round that contradict their beliefs, which drives them to switch. On the other hand, there are also agents who believe more in their beliefs and stick to the same vote. 
An example of the prompts, answers, and reasoning of agents is in Appendix~\ref{apx:exp}. 

\subsection{Take-away Messages}

Our experiments illustrate a scenario where two-round voting is more likely to lead to an informed majority decision than one-round voting played by generative AI agents. 
These findings can guide future applications of generative AI in voting systems, such as AI-augmented or AI-proxy voting. As AI reasoning capabilities continue to improve, two-round voting will likely become even more effective in reaching informed outcomes. Moreover, our experiments serve as a practical reference for designing real-world experiments that compare different voting rules and explore strategic behavior. Because generative AI simulations allow for direct preference injection, they offer a valuable tool for preliminary experimentation and design.

}

\section{Conclusion and Future Work}
We study the voting problem with two alternatives where voters' preferences depend on a not-directly-observable state variable and propose the two-round voting mechanism where the first round serves as a polling stage and the winning alternative only depends on the outcome of the second round.
Our analysis demonstrates that this mechanism can achieve informed majority decisions under strategic behavior through a natural equilibrium combining informative first-round and sincere second-round voting. This sharply contrasts one-round voting mechanisms that often fail to provide both simplicity and effective equilibria simultaneously. Furthermore, we prove that the two-round mechanism supports second-round consistent equilibria, where voter behavior mirrors that of one-round voting scenarios. These properties establish two-round voting as an elegant solution that combines simplicity with robust information aggregation capabilities, ultimately achieving informed majority decisions effectively.



A future direction is to formulate the two-round voting game as an extensive-form game.
To incorporate the analysis with group manipulations, we need to first find a reasonable equilibrium concept that extends the strong Bayes Nash equilibrium to the extensive-form game.
It would be interesting to see if the corresponding equilibria lead to good election outcomes and compare the equilibria to the ones in this paper. 


In this paper, we assume that agents share a common prior on the world state and private signals. If agents have no information about the prior, then, it has been shown in various settings that strategic agents may not be able to make an informed majority decision~\citep{Feddersen97:Voting,schoenebeck21wisdom}.  However, it is still an interesting question of whether we could relax the common knowledge assumption.  For example, what if the agents had common knowledge about the distribution of possible world/signal models?  In this case, it remains unknown whether one-round or two-round voting mechanisms can still achieve the informed majority decision. We are interested in whether the polling stages in the two-round or multi-round voting can help agents share information and form a common belief so that the correct decision can be made eventually in the last round. 

Another possible direction is to conduct human-subject experiments comparing one-round and two-round voting mechanisms, as a further step of our generative AI experiments. We would like to know both if the polling stage indeed helps real-world voters to aggregate their information and preferences more easily, and how likely and when one-round and two-round mechanisms have different outcomes.


\bibliographystyle{ACM-Reference-Format}
\bibliography{references,newref}

\clearpage
\appendix


\section{More equilibrium analysis in one-round voting}
\label{apx:one_round}
{

In this section, we give more examples where natural strategies fail to be equilibria in the one-round voting, including sincere voting (where an agent chooses the alternative maximizing his/her expected utility when this alternative becomes the winner) and the Bayesian strategy (an agent votes for an alternative with the probability of the Bayesian posterior probability that this alternative is the informed majority decision). 

We follow the example in Section~\ref{sec:one_round_unnatural} where all agents are contingent, and the signal distribution is $P_{hH}=0.9, P_{lH}=0.1, P_{hL}=0.7$, and $P_{lL}=0.3$. In Section~\ref{sec:one_round_unnatural} we have shown that any pure symmetric strategy fails to be an equilibrium in this setting. Now we show that sincere voting fails, and the Bayesian strategy is an equilibrium only when 
common prior is $P_H = P_L = 0.5$,

One natural strategy is to choose the alternative that maximizes the agent's expected utility, assuming this agent is the only voter who decides the outcome.
This is called \emph{the sincere strategy}. 
Suppose the common prior is $P_H = P_L = 0.5$. 
An agent's posterior belief for world $H$ upon receiving signal $h$ is
$$\Pr(W=H\mid S_i=h)=\frac{\Pr(W=H)\cdot\Pr(S_i=h\mid W=H)}{\Pr(S_i=h)}=\frac{0.5\times 0.9}{0.5\times 0.9+0.5\times0.7}=\frac9{16},$$
and the posterior belief for world $H$ upon receiving signal $l$ is
$$\Pr(W=H\mid S_i=l)=\frac{\Pr(W=H)\cdot\Pr(S_i=l\mid W=H)}{\Pr(S_i=l)}=\frac{0.5\times 0.1}{0.5\times 0.1+0.5\times0.3}=\frac14.$$
If an agent $i$ receives $h$, the expected utility of voting for $\bA$ is
$$\Pr(W=H\mid S_i=h)\cdot v_i(H,\bA)+\Pr(W=L\mid S_i=h)\cdot v_i(L,\bA)=\frac9{16}v_i(H,\bA)+\frac7{16}v_i(L,\bA),$$
and the expected utility of voting for $\bR$ is
$$\Pr(W=H\mid S_i=h)\cdot v_i(H,\bR)+\Pr(W=L\mid S_i=h)\cdot v_i(L,\bR)=\frac9{16}v_i(H,\bR)+\frac7{16}v_i(L,\bR).$$
When voting sincerely, the agent receiving signal $h$ will vote for $\bA$ if and only if the expected utility for $\bA$ is higher:
$$\frac9{16}v_i(H,\bA)+\frac7{16}v_i(L,\bA)>\frac9{16}v_i(H,\bR)+\frac7{16}v_i(L,\bR).$$
Similarly, the agent receiving signal $l$ will vote for $\bA$ if and only if
$$\frac14v_i(H,\bA)+\frac34v_i(L,\bA)>\frac14v_i(H,\bR)+\frac34v_i(L,\bR).$$
The utility for a contingent agent satisfies that $v_i(H,\bA)>v_i(H,\bR)$ and $v_i(L,\bR)>v_i(L,\bA)$.
However, when agents care more about the outcome under world $H$ and are relatively indifferent about the outcome under world $L$, they will always vote for $\bA$ and it is not favored by them under world $L$.
In our example, if $v_i(H,\bA)-v_i(H,\bR)>3\cdot(v_i(L,\bR)-v_i(L,\bA))$, agents will vote for $\bA$ even when signal $l$ is received.
Thus, the sincere strategy profile does not guarantee the informed majority decision, and it is not a strong BNE by the if-and-only-if characterization.

Perhaps one more natural strategy is to vote according to the posterior, i.e., adopt the mixed strategy with probabilities matching the posterior belief: an agent votes for $\bA$ with probability $\beta_{\ell} = \Pr(W=H\mid S_i=l)$ when receiving signal $\ell$ and probability $\beta_h=\Pr(W=H\mid S_i=h)$ when receiving signal $h$.
In our specific \textinst{}, with the prior belief $P_H=P_L=0.5$, our previous calculations reveal that the strategy is given by $(\beta_l,\beta_h)=(\frac14,\frac9{16})$.
When the actual world is $H$, the ``correct'' alternative $\bA$ gets a share of $0.1\cdot\frac14+0.9\cdot\frac9{16}=53.125\%$ in expectation;
when the actual world is $L$, the ``incorrect'' alternative $\bA$ wins a share of
$0.3\cdot\frac14+0.7\cdot\frac9{16}=46.875\%$ in expectation.
This looks good: the ``correct'' alternative always receives more than half of the votes in expectation.
Some calculations reveal that the informed majority decision can be guaranteed (more formally, the fidelity goes to $1$ when $n\rightarrow\infty$) by this kind of matching-posterior strategy profile if the prior belief is half-half (i.e., $P_H=P_L=0.5$).
However, if the prior belief is significantly biased, say, $P_H$ is close to $1$, some calculations reveal that the expected fraction of $\bA$ votes is close to $1$ under both world states, in which case the informed majority decision is not secured under the world state $L$. 
}


\section{Full Proofs}
\label{apx:proof}
{

We abuse the notation for deterministic first/second round strategy such that $\stgone_{\sag}(\sig) = \bA$ means that $\sag$ always votes for $\bA$ (with probability $1$) in the first round when her signal is $\sig$. Other deterministic votes follow a similar notation.

The proofs are written under the general valuation setting and can be transferred to those under 0-1 utility setting by assigning the utilities. 


\subsection{Proof of Lemma~\ref{lem:deviate}}
\label{apx:deviate}

\begin{lembis}{lem:deviate}[Formal]
    Let $\stgp^*$ be an arbitrary profile, and let $\varepsilon = 2B(B+2)(1 - \acc(\stgp^*))$. If there exists a group of agents $D$ and another strategy profile $\stgp'$ such that agents in $D$ have incentives to deviate to $\stgp'$ so that (1) $\forall$ $i\in D$, $\ut_{\sag}(\stgp') \ge \ut_{\sag}(\stgp^*)$, and (2) there exists $\sag\in D$ such that $\ut_{\sag}(\stgp') > \ut_{\sag}(\stgp) +\varepsilon$, 
    then either $D \subseteq F$ or $D \subseteq U$. 
\end{lembis}


    In this proof, we use $\stgp^*$ to denote $\stgp^*_{\ag}$ for simplicity. We show that for any other strategy profile $\stgp'$, there does not exist a group of agents wishing to deviate. Let $e = 1 - \acc(\stgp^*)$. We know that $e$ converges to $0$ as $\ag$ goes to $1$. We take $\varepsilon = 2B(B+2)e$. 
    

\begin{proof}[Proof of Lemma~\ref{lem:deviate}]
    We consider two different cases of the fidelity of $\stgp'$.

    \paragraph{Case 1: $(1 - \acc(\stgp')) < (2B + 2)e$.} In this case, since the fidelity of $\stgp'$ is also close to one, two profiles make no significant difference, and no agent can gain more than $\varepsilon = 2B(B+2)e$.

    Consider the difference of expected utility in two profiles. 

    \begin{align*}
        \ut_{\sag}(\stgp') - \ut_{\sag}(\stgp^*) = &\ P_H\cdot (\lp_{H}^{\bA}(\stgp') - \lp_{H}^{\bA}(\stgp^*))(\vt_{\sag}(H, \bA) - \vt_{\sag}(H,\bR))\\
        &\ + P_L\cdot (\lp_{L}^{\bR}(\stgp') - \lp_{L}^{\bR}(\stgp^*))(\vt_{\sag}(L, \bR) - \vt_{\sag}(L,\bA)).
    \end{align*}

    Since we know that $\acc(\stgp^*)=1 - e$, we have $\lp_{H}^{\bA}(\stgp^*) \ge 1 - \frac{e}{P_H}$ and $\lp_{L}^{\bR}(\stgp^*) \ge 1 - \frac{e}{P_L}$. Similarly, for $\stgp'$, we have $\lp_{H}^{\bA}(\stgp') \ge 1 - \frac{(2B+2)e}{P_H}$ and $\lp_{L}^{\bR}(\stgp') \ge 1 - \frac{(2B+2)e}{P_L}$. Then we can bound the difference between these probabilities.
    \begin{align*}
        \lp_{L}^{\bR}(\stgp^*) - \lp_{L}^{\bR}(\stgp') \le &\ \frac{(2B+2)e}{P_L},\\
        \lp_{L}^{\bR}(\stgp') - \lp_{L}^{\bR}(\stgp^*) \le &\ \frac{e}{P_L},\\
        \lp_{H}^{\bA}(\stgp^*) - \lp_{H}^{\bA}(\stgp') \le &\ \frac{(2B+2)e}{P_H},\\
        \lp_{H}^{\bA}(\stgp') - \lp_{H}^{\bA}(\stgp^*) \le &\ \frac{e}{P_H}.
    \end{align*}
    Then we can bound the expected utilities.
    \begin{itemize}
        \item For the friendly agents, $\vt_{\sag}(H, \bA) > \vt_{\sag}(H, \bR)$ and $\vt_{\sag}(L, \bA) > \vt_{\sag}(L, \bR)$. Then 
        \begin{align*}
            \ut_{\sag}(\stgp') - \ut_{\sag}(\stgp^*) \le &\ P_L\cdot \frac{(2B+2)e}{P_L}\cdot B + P_H \cdot \frac{e}{P_H}\cdot B\\
            =&\ 2B(B+2)e.
        \end{align*}
        \item For contingent agents, $\vt_{\sag}(H, \bA) > \vt_{\sag}(H, \bR)$ and $\vt_{\sag}(L, \bA) < \vt_{\sag}(L, \bR)$. Then 
        \begin{align*}
            \ut_{\sag}(\stgp') - \ut_{\sag}(\stgp^*) \le &\ P_L\cdot \frac{e}{P_L}\cdot B + P_H \cdot \frac{e}{P_H}\cdot B\\
            =&\ 2Be.
        \end{align*}
        \item For the unfriendly agents, we have $\vt_{\sag}(H, \bA) < \vt_{\sag}(H, \bR)$ and $\vt_{\sag}(L, \bA) < \vt_{\sag}(L, \bR)$. Then 
        \begin{align*}
            \ut_{\sag}(\stgp') - \ut_{\sag}(\stgp^*) \le &\ P_L\cdot \frac{e}{P_L}\cdot B + P_H \cdot \frac{(2B+2)e}{P_H}\cdot B\\
            =&\ 2B(B+2)e.
        \end{align*}
    \end{itemize} 
    Therefore, no agents can gain more than $\varepsilon = 2B(B+2)e$.

    \paragraph{Case 2: $(1 - \acc(\stgp'))\ge (2B+2)e$. } In this case, we proceed with our proof in two steps. 
    \begin{enumerate}
        \item Firstly, contingent agents will not deviate, because their expected utility will strictly decrease. 
        \item Secondly, a deviating group cannot contain both friendly and unfriendly agents.  
    \end{enumerate}

    \begin{claim}
        \label{claim:contingent}
       If $(1 - \acc(\stgp'))\ge (2B+2)e$, then for any contingent agent $\sag$, $\ut_{\sag}(\stgp') - \ut_{\sag}(\stgp^*) < 0$. 
    \end{claim}
    Still, recall the difference between the two expected utilities. 
    \begin{align*}
        \ut_{\sag}(\stgp') - \ut_{\sag}(\stgp^*) = &\ P_H\cdot (\lp_{H}^{\bA}(\stgp') - \lp_{H}^{\bA}(\stgp^*))(\vt_{\sag}(H, \bA) - \vt_{\sag}(H,\bR))\\
        &\ + P_L\cdot (\lp_{L}^{\bR}(\stgp') - \lp_{L}^{\bR}(\stgp^*))(\vt_{\sag}(L, \bR) - \vt_{\sag}(L,\bA)).
    \end{align*}

    From $\acc(\stgp^*) = e$ and $\acc(\stgp') \le 1 - (2B+2)e$, we have 
    \begin{equation*}
        P_{L}\cdot \lp_{L}^{\bR}(\stgp^*) + P_{H} \cdot \lp_{H}^{\bA}(\stgp^*) -  P_{L}\cdot \lp_{L}^{\bR}(\stgp') + P_{H} \cdot \lp_{H}^{\bA}(\stgp') \ge (2B+1)e. 
    \end{equation*}
    Then we consider different cases between the $\lambda$. 
    
    \begin{itemize}
        \item If $\lp_{H}^{\bA}(\stgp') \le \lp_{H}^{\bA}(\stgp^*)$ and $\lp_{L}^{\bR}(\stgp') \le \lp_{L}^{\bR}(\stgp^*)$, at least one of them will be strict. Then $\ut_{\sag}(\stgp') - \ut_{\sag}(\stgp^*) < 0$. 
        \item If $\lp_{H}^{\bA}(\stgp') > \lp_{H}^{\bA}(\stgp^*)$, then we have $\lp_{H}^{\bA}(\stgp') -\lp_{H}^{\bA}(\stgp^*) < \frac{e}{P_H}$. On the other hand, $\lp_{L}^{\bR}(\stgp') - \lp_{L}^{\bR}(\stgp^*)\le -\frac{(2B+1)e}{P_L}$. Then, 
        \begin{align*}
            \ut_{\sag}(\stgp') - \ut_{\sag}(\stgp^*) \le P_H\cdot \frac{e}{P_H}\cdot B + P_L \cdot \left(-\frac{(2B+1)e}{P_L}\right)\cdot 1 < 0.
        \end{align*}
        \item If $\lp_{L}^{\bR}(\stgp') > \lp_{L}^{\bR}(\stgp^*)$, with similar reasoning, we can show that $\ut_{\sag}(\stgp') - \ut_{\sag}(\stgp^*) < 0$. 
    \end{itemize}
    
    \begin{claim}
        \label{claim:predetermined}
        Let $\sag_1$ be any friendly agent and $\sag_2$ be any unfriendly agent. Then 
        \begin{enumerate}
            \item If $\ut_{\sag_1}(\stgp') - \ut_{\sag_1}(\stgp^*) > \varepsilon$, then $\ut_{\sag_2}(\stgp') - \ut_{\sag_2}(\stgp^*) < 0$.
            \item If $\ut_{\sag_2}(\stgp') - \ut_{\sag_2}(\stgp^*) > \varepsilon$, then $\ut_{\sag_1}(\stgp') - \ut_{\sag_1}(\stgp^*) < 0$.
        \end{enumerate}
    \end{claim}
    We only prove 1, and 2 will follow the same reasoning. Without loss of generality, let $\sag_1 = 1$ and $\sag_2 = 2$. 

    Still, recall the difference between the expected utility. 
    \begin{align*}
        \ut_{\sag}(\stgp') - \ut_{\sag}(\stgp^*) = &\ P_H\cdot (\lp_{H}^{\bA}(\stgp') - \lp_{H}^{\bA}(\stgp^*))(\vt_{\sag}(H, \bA) - \vt_{\sag}(H,\bR))\\
        &\ + P_L\cdot (\lp_{L}^{\bR}(\stgp') - \lp_{L}^{\bR}(\stgp^*))(\vt_{\sag}(L, \bR) - \vt_{\sag}(L,\bA)).
    \end{align*}
    We consider two cases. 
    
    Case 1: $\lp_{H}^{\bA}(\stgp') \ge \lp_{H}^{\bA}(\stgp^*)$. In this case, since $\lp_{H}^{\bA}(\stgp^*)$ is already close to 1, the gain in agent $1$'s expected utility on $H$'s side will not exceed $P_H\cdot \frac{e}{P_H}\cdot B = eB < \varepsilon$. Therefore, $\lp_{L}^{\bR}(\stgp') <\lp_{L}^{\bR}(\stgp^*)$. In this case, the expected utility of agent 2 will strictly decrease. 

    Case 2: $\lp_{H}^{\bA}(\stgp') < \lp_{H}^{\bA}(\stgp^*) $. In this case, the $P_L$ term is positive for agent $1$, while the $P_H$ term is positive for agent 2. For agent $1$, we have $\vt_{1}(H, \bA) \ge \vt_{1}(L, \bA) > \vt_{1}(L, \bR) \ge \vt_{1}(H, \bR)$. Then since $\ut_{1}(\stgp') - \ut_{1}(\stgp^*)$ is strictly positive, we must have 
    \begin{equation*}
         P_H\cdot (\lp_{H}^{\bA}(\stgp') - \lp_{H}^{\bA}(\stgp^*)) > P_L\cdot (\lp_{L}^{\bR}(\stgp') - \lp_{L}^{\bR}(\stgp^*)).
    \end{equation*}

    For agent $2$, on the other hand, we have $\vt_{1}(L, \bR) \ge \vt_{1}(H, \bR) > \vt_{1}(H, \bA) \ge \vt_{1}(L, \bA)$. Then there must have $\ut_{2}(\stgp') - \ut_{2}(\stgp^*) < 0$.
\end{proof}

\subsection{Proof of Theorem~\ref{thm:threshold}}
\label{apx:threshold}

The proof of Theorem~\ref{thm:threshold} resembles that of Theorem~\ref{thm:sincere}. Lemma~\ref{lem:fidelity_threshold} guarantees the the fidelity. For the equilibrium parts, applying Lemma~\ref{lem:deviate}, it suffice to consider a deviating group with only friendly agents or only unfriendly agents. Finally, Claim~\ref{claim:oneside_threshold} shows that such deviation cannot succeed. 

\begin{lembis}{lem:fidelity_threshold}
Under any \textinst{} sequence $\{\instance_{\ag}\}$, every constantly separable sequence of profiles $\{\stgp^*\}$ satisfies that (1) $\lim_{\ag\to\infty}\acc(\stgp^{\dag}_{\ag}) = 1$,  and (2) for every $\ag$, $\stgp^{*}_{\ag}$ is a $\varepsilon$-strong BNE where $\varepsilon$ converges to 0 as $\ag$ goes to $\infty$. 
\end{lembis} 

\begin{proof}
    The high-level idea of $\stgp^*$ is that contingent agents determine the world state with high probability via the first-round voting and vote for the correct alternative in the second round. Since predetermined agents behave regularly, and both $\agf$ and $\agu$ do not exceed 0.5, the alternative that all contingent agents vote for will be the winner. Now we show that both $\lim_{\ag\to\infty} \lp_{H}^{\bA}(\stgp^*) = 1$ and $\lim_{\ag\to\infty} \lp_{L}^{\bR}(\stgp^*) = 1$.

    For now we fix an arbitrary $\ag$.
    For an agent $\sag$, Let $X_{\sag}$ be the random variable indicating $\sag$'s vote in the first round. $X_{\sag} = 0$ means $\sag$ votes for $\bR$, and $X_{\sag} = 1$ means $\sag$ votes for $\bA$. 
    Then the result of the first round voting $\hst$ can be represented as 
    $\hst = \sum_{\sag} X_{\sag}. $ 
    If the first round outcome exceeds $\hst \ge \tf + \tc \cdot P_{hH} - \delta$, then all the agents $i$ will vote for $\bA$, and $\bA$ becomes the winner. 
    Therefore, the probability such that $\bA$ becomes the winner  conditioned on  the world state $H$ satisfies

    \begin{equation*}
        \lambda_H^{\bA}(\stgp^*) \ge Pr\left[\sum_{\sag} X_{\sag}\ge \tf + \tc \cdot P_{hH} - \delta\cdot \ag\ \middle\vert\  H\right].
    \end{equation*}

Since contingent agent votes informatively in the first round, the expectation of of the first round outcome $E[\sum_{\sag} X_{\sag} \mid H] = \tf + \tc \cdot P_{hH}$.
Then for all $\ag \ge \ag_0$, we can apply the Hoeffding inequality to the probability:
\begin{align*}
   \lambda_H^{\bA}(\stgp^*) \ge&\ \Pr\left[\sum_{\sag} X_{\sag}\ge \tf + \tc \cdot P_{hH} - \delta\cdot \ag\ \middle\vert\  H\right]\\
     = &\ \Pr\left[\sum_{\sag} X_{\sag} -E[\sum_{\sag} X_{\sag} \mid H] \ge \tf + \tc \cdot P_{hH} - \delta \cdot \ag- E[\sum_{\sag} X_{\sag} \mid H] \ \middle\vert\  H\right]\\
    = &\ \Pr\left[\sum_{\sag} X_{\sag} - E[\sum_{\sag} X_{\sag} \mid H]\ge -\delta\cdot \ag \ \middle\vert\ H\right]\\
    \ge &\ 1 - \exp(2\delta^2\cdot \ag)). 
\end{align*}

Similarly, when the world state is $L$, the expectation of the first-round outcome is $E[\sum_{\sag} X_{\sag} \mid L = \tf + \tc \cdot P_{hL}$. If the first round outcome is at most $\hst \ge \tf + \tc \cdot P_{hL} + \delta$, then all the agents $i$ will vote for $\bR$, and $\bR$ becomes the winner. 
Therefore, the probability such that $\bR$ becomes the winner  conditioned on  the world state $L$ satisfies
\begin{align*}
   \lambda_L^{\bR}(\stgp^*) \ge&\ \Pr\left[\sum_{\sag} X_{\sag}\le \tf + \tc \cdot P_{hL} + \delta\cdot \ag\ \middle\vert\  L\right]\\
     = &\ \Pr\left[\sum_{\sag} X_{\sag} -E[\sum_{\sag} X_{\sag} \mid L] \le \tf + \tc \cdot P_{hL} + \delta \cdot \ag- E[\sum_{\sag} X_{\sag} \mid L] \ \middle\vert\  L\right]\\
    = &\ \Pr\left[\sum_{\sag} X_{\sag} - E[\sum_{\sag} X_{\sag} \mid L]\ge \delta\cdot \ag \ \middle\vert\ L\right]\\
    \ge &\ 1 - \exp(2\delta^2\cdot \ag)). 
\end{align*}

Therefore, both $\lp_{H}^{\bA}(\stgp^*_{\ag})$ and  $\lp_{L}^{\bR}(\stgp^*_{\ag})$ converge to 1 as $\ag$ goes to infinity. Since fidelity $\acc(\stgp^*) = P_H\cdot\lp_{H}^{\bA}(\stgp^*_{\ag}) + P_L\cdot \lp_{L}^{\bR}(\stgp^*_{\ag})$, $\lim_{\ag\to\infty} \acc(\stgp^*) = 1$. 
\end{proof}

\begin{lembis}{lem:eq_threshold}
    $\stgp^{*}_{\ag}$ is an $\varepsilon$-strong BNE with $\varepsilon = o(1)$. 
\end{lembis}

By applying Lemma~\ref{lem:deviate}, we only need to show that the deviating group containing only one type of predetermined agents cannot succeed. Therefore, ``informative + threshold'' strategy profile $\stgp^*$ is an $\varepsilon$-strong BNE with $\varepsilon$ converges to 0. 

\begin{claim}
        \label{claim:oneside_threshold}
        Suppose $D$ is the deviating group that contains only friendly or only unfriendly agents, then $\ut_{\sag}(\stgp') - \ut_{\sag}(\stgp^*) < \varepsilon$ for all $\sag \in D$.
    \end{claim}
    
We prove the friendly agent side, and the unfriendly agent side will follow the same reasoning. The expected utility of a friendly agent increases only when the likelihood that $\bA$ becomes the winner increases. However, in $\stgp^*_\ag$ , all the friendly agents always vote for $\bA$ in both rounds. If friendly agents deviate in the second round, there will be strictly fewer agents voting for $\bA$ in expectation, and $\bA$ will be less likely to be the winner. If friendly agents deviate in the first round, the result of the first-round voting $x$ will have a higher probability of being smaller, and contingent agents playing the threshold strategy will tend to vote $\bR$. Therefore, no matter how $D$ deviates, the likelihood of $\bA$ being the winner will not increase, and friendly agents cannot get a gain of more than $\varepsilon$.

\subsection{Proof for Lemma~\ref{lem:sincere_is_threshold}}
\begin{lembis}{lem:sincere_is_threshold}
    The sequence of ``informative + sincere'' strategy profiles is constantly separable. 
\end{lembis}

\begin{proof}
    Recall that 
    $$\hst ^* = \tf + \frac{\log\frac{P_L \cdot (\vt_i(L, \bR) - \vt_i(L, \bA))}{P_H \cdot (\vt_i(H, \bA) - \vt_i(H, \bR))} + \tc \cdot \log \frac{P_{\ell L}}{P_{\ell H}}}{\log \frac{P_{hH}P_{\ell L}}{P_{\ell H}P_{hL}}}. $$
As $\tf$ and $\tc$ increases with $\ag$, while $\log\frac{P_L \cdot (\vt_i(L, \bR) - \vt_i(L, \bA))}{P_H \cdot (\vt_i(H, \bA) - \vt_i(H, \bR))}$ remains constant, it suffice to show that 
\begin{equation*}
    \tf + \tc \cdot P_{hL} + \delta \cdot \ag \le \tf + \frac{\tc \cdot \log \frac{P_{\ell L}}{P_{\ell H}}}{\log \frac{P_{hH}P_{\ell L}}{P_{\ell H}P_{hL}}} \le \tf + \tc \cdot P_{hH} - \delta \cdot \ag,
\end{equation*}
which is equivalent to
\begin{equation*}
    P_{hL} + \delta \cdot \frac{\ag}{\tc} \le \frac{\log \frac{P_{\ell L}}{P_{\ell H}}}{\log \frac{P_{hH}P_{\ell L}}{P_{\ell H}P_{hL}}} \le P_{hH}- \delta \cdot\frac{\ag}{\tc}. 
\end{equation*}
    
\end{proof}

We first deal with the left inequality. 
\begin{align*}
    &\ P_{hL} + \delta \cdot \frac{\ag}{\tc} \le \frac{\log \frac{P_{\ell L}}{P_{\ell H}}}{\log \frac{P_{hH}P_{\ell L}}{P_{\ell H}P_{hL}}}\\
    \Leftrightarrow &\  P_{hL} + \delta \cdot \frac{\ag}{\tc} \le 
    \frac{\log \frac{P_{\ell L}}{P_{\ell H}}}{\log \frac{P_{hH}}{P_{hL}} + \log \frac{P_{\ell L}}{P_{\ell H}} }\\
    \Leftrightarrow &\ \delta \cdot \frac{\ag}{\tc} \cdot (\log \frac{P_{hH}}{P_{hL}} + \log \frac{P_{\ell L}}{P_{\ell H}} ) \le 
    (P_{hL}\cdot \log \frac{P_{hL}}{P_{hH}} + \log P_{\ell L} \cdot \frac{P_{\ell L}}{P_{\ell H}}).
\end{align*}
Note that the RHS takes a form of the KL-divergence between distribution $(P_{hL}, P_{\ell L})$ and $P_{hH}, P_{\ell H})$. Given that $P_{hH} > P_{hL}$, RHS is a strictly positive constant independent from $\ag$. On the other hand, $\tc \approx \agc \cdot \ag$. Therefore, by taking 
\begin{equation*}
    \delta \le \frac12 \agc\cdot  \frac{(P_{hL}\cdot \log \frac{P_{hL}}{P_{hH}} + \log P_{\ell L} \cdot \frac{P_{\ell L}}{P_{\ell H}})}{(\log \frac{P_{hH}}{P_{hL}} + \log \frac{P_{\ell L}}{P_{\ell H}} )} ,
\end{equation*}
the left inequality is satisfied. 

Likewise, the right inequality is equivalent to
\begin{equation*}
    \delta \cdot \frac{\ag}{\tc} \cdot (\log \frac{P_{hH}}{P_{hL}} + \log \frac{P_{\ell L}}{P_{\ell H}} ) \le 
    (P_{hH}\cdot \log \frac{P_{hH}}{P_{hL}} + \log P_{\ell H} \cdot \frac{P_{\ell H}}{P_{\ell L}}),
\end{equation*}
and by taking
\begin{equation*}
    \delta \le \frac12 \agc\cdot  \frac{(P_{hH}\cdot \log \frac{P_{hH}}{P_{hL}} + \log P_{\ell H} \cdot \frac{P_{\ell H}}{P_{\ell L}})}{(\log \frac{P_{hH}}{P_{hL}} + \log \frac{P_{\ell L}}{P_{\ell H}} )},
\end{equation*}
the right inequality is satisfied. Therefore, by a suffciently small $\delta$ and all sufficiently large $\ag$, the constaint for constantly separable is satisfied.

\subsection{Proof of Theorem~\ref{thm:eqtoacc}}

The proof applies to both the 0-1 utility setting (Theorem~\ref{thm:eqtoacc}) and the general valuation utility setting (Theorem~\ref{thm:eqtoacc_nb}
). 
Before starting the proof, we extend the definition of the first-round voting share. let $\exs_H^{F}(\stgp), \exs_H^{U}(\stgp)$, and $\exs_H^{C}(\stgp)$ be the expectation on the number of votes for $\bA$ from three types of agents conditioned on the world state is $H$ and the strategy profile is $\stgp$. We have $\exs_H(\stgp) = \exs_H^{F}(\stgp) + \exs_H^{U}(\stgp) + \exs_H^{C}(\stgp)$. The $L$ side is defined similarly. 

\label{apx:eqtoacc}
\begin{proof}
    We first give the proof for regular profiles. 
    Suppose $\{\stgp_{\ag}\}$ is a regular profile sequence such that $\lim_{\ag\to\infty} \acc(\stgp_\ag) = 1$ does not hold. 
    This means that there exist a constant $\delta$ and an infinite set $\mathcal{N}$ such that for all $\ag \in \mathcal{N}$, $\acc(\stgp_{\ag}) \le 1 - \delta$. We show that there exists a constant $\varepsilon = \frac{\delta\cdot P_H P_L}{2}$ and a infinite set $\mathcal{N}'$ such that for all $\ag\in \mathcal{N}'$, $\stgp_{\ag}$ is NOT an $\varepsilon$-strong BNE. 


    We will first construct a different profile sequence $\{\stgp'_{\ag}\}$ whose fidelity converges to 1. Then for all sufficiently large $\ag \in \mathcal{N}$, both $\lp_{H}^{\bA}$ and $\lp_{L}^{\bR}$ are close to one. On the other hand, Since $\acc(\stgp_{\ag}) \le 1 - \delta$, we have at least one of $\lp_{H}^{\bA}(\stgp_{\ag}) \le 1 - \delta$ and $\lp_{L}^{\bR}(\stgp_{\ag}) \le 1-\delta$ holds. Then we can see that all the contingent agents have incentives to deviate to $\stgp'_{\ag}$ to get a constant gain in the expected utility.

    For simplicity, we use $\stgp$ and $\stgp'$ to denote $\stgp_\ag$ and $\stgp'_\ag$. Let $\stgp'_{\ag}$ be the following strategy profile. The deviating group $D$ contains all the contingent agents. 
    
    

    \begin{itemize}
        \item All predetermined agents play the same strategy as in $\stgp$. 
        \item If $\exs_H^{F}(\stgp) + \exs_H^{U}(\stgp) \ge \exs_L^{F}(\stgp) + \exs_L^{U}(\stgp)$, all contingent agents vote informatively in the first round. For the second round, all contingent agents vote for $\bA$ if $\hst \ge \frac{\exs_H(\stgp') + \exs_L(\stgp')}{2}$ and vote for $\bR$ otherwise.
        \item If  $\exs_H^{F}(\stgp) + \exs_H^{U}(\stgp) < \exs_L^{F}(\stgp) + \exs_L^{U}(\stgp)$, contingent agents ``reverse'' their strategy. In the first round, a contingent agent votes for $\bA$ if receiving $l$ and for $\bR$ if receiving $h$. And in the second round,  all contingent agents vote for $\bA$ if $\hst < \frac{\exs_H(\stgp') + \exs_L(\stgp')}{2}$ and vote for $\bR$ otherwise.
    \end{itemize}

    The high-level idea of $\stgp'$ is similar to $\stgp^*$ in Theorem~\ref{thm:threshold}: contingent agents determine the world state with high probability via the first-round voting and vote for the correct alternative in the second round. Since predetermined agents behave regularly, and both $\agf$ and $\agu$ do not exceed 0.5, the alternative that all contingent agents vote for will be the winner. Now we show that both $\lim_{\ag\to\infty} \lp_{H}^{\bA}(\stgp') = 1$ and $\lim_{\ag\to\infty} \lp_{L}^{\bR}(\stgp') = 1$.

    With loss of generality, suppose $\exs_H^{F}(\stgp) + \exs_H^{U}(\stgp) \ge \exs_L^{F}(\stgp) + \exs_L^{U}(\stgp)$. The reasoning for the other case will be similar. For an agent $\sag$, Let $X_{\sag}$ be the random variable indicating $\sag$'s vote in the first round. $X_{\sag} = 0$ means $\sag$ votes for $\bR$, and $X_{\sag} = 1$ means $\sag$ votes for $\bA$. 
    Then the result of the first round voting $\hst$ can be represented as 
    $\hst = \sum_{\sag} X_{\sag}. $ 
    Then the probability such that $\bA$ becomes the winner (which is the probability that $\hst \ge \frac{\exs_H(\stgp') + \exs_L(\stgp')}{2}$) conditioned on  the world state $H$ can be written as 
    \begin{equation*}
        \lambda_H^{\bA}(\stgp') = \Pr[\sum_{\sag} X_{\sag}\ge \frac{\exs_H(\stgp') + \exs_L(\stgp')}{2} \mid H].
    \end{equation*}

    Note that $E[\sum_{\sag} X_{\sag} \mid H] = \exs_H(\stgp')$. 
    
    Since contingent agent votes informatively in the first round, $\exs_H^{C}(\stgp') = \tc \cdot P_{hH}$ and $\exs_L^{C}(\stgp') = \tc \cdot P_{hL}$. Since $P_{hH} > P_{hL}$, we have $\exs_H^{C}(\stgp') - \exs_L^{C}(\stgp') = \tc \cdot (P_{hH} - P_{hL}) \ge \ag \cdot (\agc - \frac{1}{\ag})\cdot (P_{hH} - P_{hL})$.  Then there exists a $\ag_0$ such that for all $\ag > \ag_0$, $\exs_H^{C}(\stgp') - \exs_L^{C}(\stgp') \ge \frac12 \ag \cdot \agc\cdot (P_{hH} - P_{hL})$. Then since $\exs_H^{F}(\stgp) + \exs_H^{U}(\stgp) \ge \exs_L^{F}(\stgp) + \exs_L^{U}(\stgp)$, we have  $\ag > \ag_0$, $\exs_H(\stgp') - \exs_L(\stgp') \ge \frac12 \ag \cdot \agc\cdot (P_{hH} - P_{hL})$. Then we can apply the Hoeffding inequality to the probability:
    \begin{align*}
        \lambda_H^{\bA}(\stgp')
        = &\ \Pr\left[\sum_{\sag} X_{\sag}\ge \frac{\exs_H(\stgp') + \exs_L(\stgp')}{2} \ \middle\vert\  H\right]\\
         = &\ \Pr\left[\sum_{\sag} X_{\sag} -E[\sum_{\sag} X_{\sag} \mid H] \ge \frac{\exs_H(\stgp') + \exs_L(\stgp')}{2} - E[\sum_{\sag} X_{\sag}  \mid H] \ \middle\vert\  H\right]\\
        = &\ \Pr\left[\sum_{\sag} X_{\sag} -\exs_H(\stgp')\ge \frac{\exs_L(\stgp') -\exs_H(\stgp')}{2} \ \middle\vert\ H\right ]\\
        \ge &\ \Pr[\sum_{\sag} X_{\sag} -\exs_H(\stgp')\ge - \frac14 \ag \cdot \agc\cdot (P_{hH} - P_{hL}) \mid H]\\
        \ge &\ 1 - \exp(-\frac18 (\agc\cdot (P_{hH} - P_{hL})^2 \cdot \ag)). 
    \end{align*}

    The $L$ side works similarly. 

    Therefore, both $\lp_{H}^{\bA}(\stgp'_{\ag})$ and  $\lp_{L}^{\bR}(\stgp'_{\ag})$ converge to 1 as $\ag$ goes to infinity. 

    Next, we show that contingent agents have incentives to deviate from $\stgp_{\ag}$ to $\stgp'_{\ag}$. Let $\mathcal{N}'$ be the set of all $\ag \in \mathcal{N}$ such that both $\lp_{H}^{\bA}(\stgp'_{\ag})\ge 1 - \frac{\delta\cdot P_H P_L}{2B}$ and  $\lp_{L}^{\bR}(\stgp'_{\ag})\ge 1 - \frac{\delta \cdot P_H P_L}{2B}$ holds. We know that $\mathcal{N}'$ is also an infinite set. We fix an arbitrary $\ag \in \mathcal{N'}$. Without loss of generality, suppose $\lp_{H}^{\bA}(\stgp_{\ag}) \le 1 - \delta$. Then we consider the difference in the expected utility of a contingent agent. 
    \begin{align*}
        \ut_{\sag}(\stgp'_{\ag}) - \ut_{\sag}(\stgp_{\ag}) = &\ P_H\cdot (\lp_{H}^{\bA}(\stgp'_{\ag}) - \lp_{H}^{\bA}(\stgp_{\ag}))(\vt_{\sag}(H, \bA) - \vt_{\sag}(H,\bR))\\
        &\ + P_L\cdot (\lp_{L}^{\bR}(\stgp'_{\ag}) - \lp_{L}^{\bR}(\stgp_{\ag}))(\vt_{\sag}(L, \bR) - \vt_{\sag}(L,\bA)).
    \end{align*}
    By assigning the bounds above, we have 
    \begin{align*}
        \ut_{\sag}(\stgp'_{\ag}) - \ut_{\sag}(\stgp_{\ag}) \ge &\ P_H \cdot (1 - \frac{\delta\cdot P_H P_L}{2B} - 1 + \delta) \cdot 1 + P_L\cdot (-\frac{\delta\cdot P_H P_L}{2B})\cdot B\\
        \ge &\ P_H\cdot \delta - \frac{\delta\cdot P_H P_L}{2}\\
        > &\ \frac{\delta\cdot P_H P_L}{2}. 
    \end{align*}
    The similar reasnoning can be achieved for $\lp_{L}^{\bR}(\stgp_{\ag}) \le 1 - \delta$. 

    Therefore, all the contingent agents have incentives to deviate from $\stgp_{\ag}$ to $\stgp'_{\ag}$ to get a constant gain. This means that $\stgp_{\ag}$ is NOT an $\varepsilon$-strong BNE for some constant $\varepsilon$. 

\subsubsection{Proof for non-regular profiles.} 
We adopt the same idea to construct a strategy profile where the deviating group forces the state to be revealed by the first-round vote and votes for the preferred agent in the second round. Then we show that as the fidelity converges to 1, the group has incentives to do so. 

For non-regular profiles, a deviating group with only contingent agents may not work, as there is no assumption on the behavior of predetermined agents in the second round. (For example, if  $\agc < 0.5$, and all the predetermined agents always vote for $\bA$, then the winner will always be $\bA$ regardless of the votes from contingent agents.) Fortunately, the following lemma shows that if the fidelity of deviating strategy profile is close to 1, then either all the friendly agents get an expected utility increase or all the unfriendly agents get an expected utility increase. Moreover, which type gets utility gain does not depend on the deviating strategy $\stgp'_\ag$. 

\begin{lembis}{lem:positive}
    Let $\{\stgp_{\ag}\}$ be a sequence of profiles such that there exists a constant $\delta$ and an infinite set $\mathcal{N}$ such that for all $\ag\in \mathcal{N}$, $\acc(\stgp_{\ag}) \le 1 - \delta$. Then there exists an infinite set $\mathcal{N}' \subseteq \mathcal{N}$ such that for each $\ag \in \mathcal{N}'$, at least one of the following holds for all sequences of profiles $\{\stgp_{\ag}'\}$ with $\lim_{\ag\to\infty} \acc(\stgp') = 1$: 
    \begin{enumerate}
        \item For any friendly agent $\sag_1$, $\ut_{\sag_1} (\stgp_{\ag}') - \ut_{\sag_1} (\stgp_{\ag}) \ge 0$. 
        \item For any unfriendly agent $\sag_2$, $\ut_{\sag_2} (\stgp_{\ag}') - \ut_{\sag_2} (\stgp_{\ag}) \ge 0$. 
    \end{enumerate}
    Moreover, for a fixed $\ag$ and $\stgp_{\ag}$, either 1 holds for all $\stgp_{\ag}'$ or 2 holds for all $\stgp_\ag'$. 
\end{lembis}
The proof of Lemma~\ref{lem:positive}
comes from the fact that a predetermined agent has different sensitivities to utility changes in different states of the world. For a friendly agent $\sag_1$, $\vt_{\sag_1}(H,\bA) > \vt_{\sag_1}(L, \bA) > \vt_{\sag_1}(L, \bR) > \vt_{\sag_1}(H, \bR)$. Therefore, $\sag_1$ is more sensitive to utility change in world state $H$ than that in world state $L$. On the other hand, an unfriendly agent $\sag_2$ is more sensitive in world state $L$. This difference in sensitivity guarantees that when the fidelity increases, the expected utility of both types of agents cannot decrease at the same time. 

\begin{proof}[Proof of Lemma~\ref{lem:positive}]
    Without loss of generality, suppose  $\ut_{\sag_1} (\stgp_{\ag}') - \ut_{\sag_1} (\stgp_{\ag}) \le 0$ holds for $\sag_1$ (which means (1) does not hold). We show that (2) $\ut_{\sag_2} (\stgp_{\ag}') - \ut_{\sag_2} (\stgp_{\ag}) > 0$ must hold for $\sag_2$. Recall that the expected utility of an agent $i$ is defined as follows. 
    \begin{align*}
    \ut_{\sag}(\stgp) = &\ P_{L}(\lp_{L}^{\bA}(\stgp)\cdot\vt_{\sag}(L, \bA) + \lp_{L}^{\bR}(\stgp)\cdot \vt_{\sag}(L, \bR))\\
    &\ +  P_{H}(\lp_{H}^{\bA}(\stgp)\cdot\vt_{\sag}(H, \bA) + \lp_{H}^{\bR}(\stgp)\cdot \vt_{\sag}(H, \bR)). 
\end{align*}
Therefore, the difference between $\ut_{\sag} (\stgp_{\ag}') $ and $\ut_{\sag} (\stgp_{\ag})$ can be represented as follows. 
 \begin{align*}
        &\ P_H\cdot (\lp_{H}^{\bA}(\stgp') - \lp_{H}^{\bA}(\stgp^*))(\vt_{\sag}(H, \bA) - \vt_{\sag}(H,\bR))\\
        &\ + P_L\cdot (\lp_{L}^{\bR}(\stgp') - \lp_{L}^{\bR}(\stgp^*))(\vt_{\sag}(L, \bR) - \vt_{\sag}(L,\bA)).
    \end{align*}

Since $\acc(\stgp_{\ag}) \le 1 - \delta$, at least one of $\lp_{H}^{\bA}(\stgp_{\ag})\le 1 - \delta$ and $\lp_{L}^{\bR}(\stgp_{\ag})\le  1- \delta$ holds. We deal with two cases separately. 

    \paragraph{Case 1: $\lp_{H}^{\bA}(\stgp_{\ag})\le 1 - \delta$.} In this case, we have $\lp_{H}^{\bA}(\stgp'_{\ag}) - \lp_{H}^{\bA}(\stgp_{\ag}) > 0$. 
    Note that for the friendly agent $\sag_1$, $\vt_{\sag_1}(H, \bA) - \vt_{\sag_1}(H,\bR) \ge -(\vt_{\sag_1}(L, \bR) - \vt_{\sag_1}(L,\bA) )$. Therefore, $\ut_{\sag_1} (\stgp') - \ut_{\sag_1} (\stgp) \le 0$ implies
    \begin{equation*}
        P_H\cdot (\lp_{H}^{\bA}(\stgp_{\ag}') - \lp_{H}^{\bA}(\stgp_{\ag})) \le P_L\cdot (\lp_{L}^{\bR}(\stgp_{\ag}') - \lp_{L}^{\bR}(\stgp_{\ag})).
    \end{equation*}
    On the other hand, for any unfriendly agent $\sag_2$, $- \vt_{\sag_2}(H, \bA) - \vt_{\sag_2}(H,\bR) \le \vt_{\sag_2}(L, \bR) - \vt_{\sag_2}(L,\bA) $. Therefore, $\ut_{\sag_2} (\stgp_{\ag}') - \ut_{\sag_2} (\stgp_{\ag}) \ge 0$. 

    \paragraph{Case 2: $\lp_{L}^{\bR}(\stgp_{\ag})\le  1- \delta$.} In this case, we have $\lp_{L}^{\bR}(\stgp_{\ag}') - \lp_{L}^{\bR}(\stgp_{\ag}) > 0$. If $\lp_{H}^{\bA}(\stgp_{\ag}') - \lp_{H}^{\bA}(\stgp_{\ag}) \le 0$, then 
    \begin{align*}
        \ut_{\sag_2} (\stgp_{\ag}') - \ut_{\sag_2} (\stgp_{\ag}) \ge  P_L\cdot (\lp_{L}^{\bR}(\stgp_{\ag}') - \lp_{L}^{\bR}(\stgp_{\ag}))(\vt_{\sag}(L, \bR) - \vt_{\sag}(L,\bA)) > 0.
    \end{align*}
    If $\lp_{H}^{\bA}(\stgp_{\ag}') - \lp_{H}^{\bA}(\stgp_{\ag}) > 0$, we apply the same reasoning as in Case 1. Therefore, $\ut_{\sag_2} (\stgp_{\ag}') - \ut_{\sag_2} (\stgp_{\ag}) \ge 0$ always holds.
For the 0-1 utility model, all these differences utilities become $1$or $-1$, and suffice to compare $P_H\cdot (\lp_{H}^{\bA}(\stgp_{\ag}') - \lp_{H}^{\bA}(\stgp_{\ag}))$ and $ P_L\cdot (\lp_{L}^{\bR}(\stgp_{\ag}') - \lp_{L}^{\bR}(\stgp_{\ag}))$ to determined which type of agents is happy. 
\end{proof}

Now we construct the deviating strategy profile sequence $\{\stgp'_{\ag}\}$, which is a combination of two strategy profile sequences $\{\stgp_F\}$ and $\{\stgp_U\}$. In both $\stgp_F$ and $\stgp_U$ contingent agents collaborate with one type of predetermined agents and deviate together. 


In $\stgp_{F}$, all the friendly and contingent agents vote informatively in the first round and play a threshold strategy in the second round.

\begin{itemize}
    \item All {\em unfriendly} agents play the same strategy as in $\stgp$. 
    \item If $\exs_H^{U}(\stgp_{\ag}) \ge  \exs_L^{U}(\stgp_{\ag})$, all friendly agents and contingent agents vote informatively in the first round. For the second round, all friendly agents and contingent agents vote for $\bA$ if $\hst \ge \frac{\exs_H(\stgp_F) + \exs_L(\stgp_F)}{2}$ and vote for $\bR$ otherwise.
    \item If  $\exs_H^{U}(\stgp_{\ag}) < \exs_L^{U}(\stgp_{\ag})$, friendly agents and contingent agents ``reverse'' their strategy. In the first round, a friendly (or contingent) agent votes for $\bA$ if receiving $l$ and for $\bR$ if receiving $h$. In the second round,  all friendly agents and contingent agents vote for $\bA$ if $\hst \ge \frac{\exs_H(\stgp_F) + \exs_L(\stgp_F)}{2}$
\end{itemize}

$\stgp_{U}$ follows the same ``informative and threshold'' pattern, except that the deviators are all the unfriendly agents and all the contingent agents. 

Since $\agf < 0.5$ and $\agu < 0.5$, if contingent agents and one type of predetermined agents vote for the same alternative, that alternative will be the winner. Therefore, with similar reasoning to the regular case, both $\stgp_{F}$ and $\stgp_{U}$ have fidelity converging to 1, and contingent agents have incentives to deviate to either of them. 

Applying Lemma~\ref{lem:positive}, we know that for each $\ag$ where the fidelity of $\stgp_\ag$ does not converge to 1, there is one type of predetermined agents whose expected utilities will increase if the deviators deviate to either $\stgp_F$ or $\stgp_U$. Moreover, the type of agents that get higher expected utilities does not depend on which deviating strategy they play. Therefore, when friendly agents get higher expected utilities, the deviators are all the friendly agents and contingent agents with $\stgp'_\ag = \stgp_F$; otherwise, the deviators are all the unfriendly agents and contingent agents with $\stgp_\ag' = \stgp_U$. Then, the expected utilities of contingent increase by at least a constant, and those of the deviating predetermined agents increase. Therefore, for each $\stgp_{\ag}$ whose fidelity is not close to 1, the deviators have incentives to deviate to either $\stgp_{F}$ or $\stgp_{U}$. Therefore, $\stgp_{\ag}$ is NOT an $\varepsilon$-strong BNE with constant $\varepsilon$.  

\end{proof}

\subsection{Formal Statement and Proof of Theorem~\ref{thm:one_and_two}}
\label{apx:one_and_two}

\begin{thmnb}{thm:one_and_two} Let $\{\stgpbar_{\ag}\}$ be a one-round regular strategy profile sequence in an \textinst{} sequence $\{\instance_{\ag} \}$. If for all $\ag$, $\stgpbar$ is an $\bar{\varepsilon}$-strong BNE in one-round voting with $\bar{\varepsilon}$ converges to 0, then for any two-round strategy profile sequence $\{\stgp_{\ag}\}$ in \textinst{} $\{\instance_{\ag}\}$, where for any $\ag$, $\stgp_{\ag}$ is a second-round consistent strategy profile of $\stgpbar_\ag$, $\stgp_{\ag}$ is an $\varepsilon$-strong BNE in the two-round voting for all $\ag$ with $\varepsilon$ converges to 0.
\end{thmnb}
\begin{proof} 
The proof proceeds in two steps. Firstly, we show that the fidelity of $\{\stgp_{\ag}\}$ converges to 1. Secondly, by applying Lemma~\ref{lem:deviate}, we know that a deviating group either contains only friendly agents or only unfriendly agents. Finally, we show that such a deviating group cannot succeed. 

\paragraph{Step 1: $\acc(\stgp_\ag)$ converges to 1.} In $\stgp_{\ag}$, the behavior and result in the first-round voting have no effect to the second-round voting that determines the winner. Agents in the second round behave exactly as they are playing a one-round voting game. Therefore, $\acc(\stgp_{\ag}) = \acc(\stgpbar_{\ag})$. Then since $\stgpbar_{\ag}$ is an $\bar{\varepsilon}$-strong BNE with $\bar{\varepsilon}$ converges to 0, we have $\acc(\stgpbar_{\ag})$ converges to 1 according to the if and only if characterization in~\citet[Theorem 2]{han2023wisdom}. Therefore, $\acc(\stgp_{\ag})$ also converges to 1.

    \paragraph{Step 2: equilibrium.} For each $\ag$, let $\varepsilon = 2B(B+2)(1 - \acc(\stgp_{\ag}))$. Fix an arbitrary $\ag$. Suppose there exists a strategy profile $\stgp'_{\ag}$ such that a deviating group $D$ of agents have incentives to deviate to $\stgp_{\ag}'$, then by Lemma~\ref{lem:deviate}, $D$ either contains only friendly agents or only unfriendly agents. 

    Then we show that such $D$ cannot succeed. This follows the last part of Lemma~\ref{lem:eq_threshold}. With loss of generality, suppose $D$ contains only friendly agents. Since $\stgpbar_{\ag}$ is a regular profile, we know that in $\stgp_{\ag}$, friendly agents always vote for $\bA$ in the second round. Then friendly agents cannot make themselves better off. The changes in the first-round voting have no effect on the second round, and changes in the second-round voting will only decrease the votes for $\bA$ and cannot increase the probability that $\bA$ becomes the winner. Therefore, for any friendly agents, the expected utility cannot increase by $\varepsilon$. 

    Therefore, we have show that for each $\ag$, $\stgp_{\ag}$ is a $\varepsilon$-strong BNE with $\varepsilon = 2B(B+2)(1 - \acc(\stgp_{\ag}))$. Since the fidelity converges to 1, $\varepsilon$ will converge to 0. 
\end{proof}}

\section{More Details on Experiments.}
\label{apx:exp}
{
\subsection{Configuration}
The experiment is run on a Lenovo Thinkbook 15 G3 AGL with AMD Ryzen 5 5600U with Radeon Graphics 2.30 GHz. The program is written in Python 3.9. 

We adopt online APIs to create generative AI agents. The API for Deepseek R1 is provided by Together AI, and the API for GPT-4o is provided by OpenAI's official API platform. 

\subsection{Examples for prompts and answers of agents.}

Here we give two examples of the prompts and the answers of the agents. More data can be found in the supplementary materials. 

\subsubsection{Deepseek R1, biased signal, accurate information, and one-round voting. }

\paragraph{Signal.} Bad Impression. 

\paragraph{Prompt.} "Consider the following hypothetical situation. Suppose you are a recruiter in a company. Your company is hiring for a new position by voting. There are in total 40 voters.You believe that the candidate should be hired only if he is qualified.Most of the company agrees with you, but a small minority <about 12\%> seem to support hiring every candidate regardless of qualifications, and another small minority <about 7\%> does not support hiring any candidate regardless of qualifications. Before meeting with the candidate, you thought she has <about 50\% chance> to be qualified. If the candidate is qualified, you will likely have a good impression after meeting with her <about 80\% chance>. If the candidate is not qualified, you will be less likely to have a good impression <about 60\% chance>.It seems that everyone is on the same page about both the candidate's likelihood of being qualified and how informative an interview is, but the interview impression is private and independent among all the voters. You got a BAD impression after meeting the candidate.After interviewing, but before any discussions, the company holds a vote. There will be only one round, where each voter vote YES or NO on whether to hire this candidate. The outcome is determined by whether the candidate gets a majority of yes votes. Analyze what is your best vote. The voting time is limited, so you cannot think everlastingly. Summarize and make an immediate decision in two minutes. Your answer should start with the analysis and end with a single character Y or N in a separate line without any other characters or symbols, indicating your vote. ",

\paragraph{Vote.} N.
\paragraph{Answer.} The candidate has a 33.3\% posterior probability of being qualified after a bad impression. With 5 automatic YES and 3 NO votes, the remaining 32 voters' expected YES votes (weighted by qualification likelihood) sum to ~21.33 + 5 = 26.33, likely securing a majority. However, given the higher probability the candidate is unqualified (66.7\%), voting NO aligns with preventing unqualified hires despite the likely majority.

\subsubsection{Deepseek R1, biased signal, vague information, and two-round voting. }
Here we give an example where an agent in two-round voting switches its vote in the second round. 

\paragraph{Signal.} Good Impression. 
\paragraph{First-round Prompt.} Consider the following hypothetical situation. Suppose you are a recruiter in a company. Your company is hiring for a new position by voting. There are in total 40 voters.You believe that the candidate should be hired only if he is qualified.Most of the company agrees with you, but a small minority seem to support hiring every candidate regardless of qualifications, and another small minority does not support hiring any candidate regardless of qualifications. Before meeting with the candidate, you thought her likelihood of being qualified was roughly balanced. If the candidate is qualified, you would almost certainly form a favorable impression after meeting her. If the candidate is not qualified, you would be far less inclined to have a positive impression, with odds comparable to a rare occurrence. It seems that everyone is on the same page about both the candidate's likelihood of being qualified and how informative an interview is, but the interview impression is private and independent among all the voters. You got a GOOD impression after meeting the candidate. After interviewing, but before any discussions, the company holds a vote.  There will be two rounds.  In each round a voter votes yes or no. In the first round, the individual votes are private, but the number of votes will be announced to all the voters. This round does not determine the outcome. In the second round, the voting is also private., the outcome is determined by whether the candidate gets a majority of YES votes. Analyze what is your best vote in the first round. The voting time is limited, so you cannot think everlastingly. Summarize and make an immediate decision in two minutes. Your answer should start with the analysis and end with a single character Y or N in a separate line without any other characters or symbols, indicating your vote.

\paragraph{First-round Vote.} Y.
\paragraph{First-round Answer.} To determine the optimal first-round vote, consider Bayesian updating: a GOOD impression strongly indicates the candidate is qualified (91\% probability). The first-round vote count informs the second-round voters. Voting Y truthfully contributes to an accurate signal, aiding rational voters in updating their beliefs. Extremists' votes are fixed, but their impact is diluted if rational voters coordinate via the first-round tally. Truthful voting maximizes collective accuracy. 

\paragraph{Second-round Prompt.} Consider the following hypothetical situation. Suppose you are a recruiter in a company. Your company is hiring for a new position by voting. There are in total 40 voters. You believe that the candidate should be hired only if he is qualified. Most of the company agrees with you, but a small minority seem to support hiring every candidate regardless of qualifications, and another small minority does not support hiring any candidate regardless of qualifications. Before meeting with the candidate, you thought her likelihood of being qualified was roughly balanced. If the candidate is qualified, you would almost certainly form a favorable impression after meeting her. If the candidate is not qualified, you would be far less inclined to have a positive impression, with odds comparable to a rare occurrence. It seems that everyone is on the same page about both the candidate's likelihood of being qualified and how informative an interview is, but the interview impression is private and independent among all the voters. You got a GOOD impression after meeting the candidate. After interviewing, but before any discussions, the company holds a vote.  There will be two rounds.  In each round, a voter votes yes or no. In the first round, the individual votes are private, but the number of votes will be announced to all the voters. This round does not determine the outcome. In the second round, the voting is also private., the outcome is determined by whether the candidate gets a majority of YES votes. After the first round, 14 of 40 voters voted yes. Analyze what is your best vote in the second round.  The voting time is limited, so you cannot think everlastingly. Summarize and make an immediate decision in two minutes. Your answer should start with the analysis and end with a single character Y or N in a separate line without any other characters or symbols, indicating your vote. 
\paragraph{Second-round Vote.} N.
\paragraph{Second-round Answer.}
After observing 14/40 first-round YES votes, Bayesian updating suggests the candidate is likely unqualified. Given rational voters' independent impressions, a low YES count aligns more with an unqualified candidate (expected ~8-12 YES) than qualified (\~28-32). Despite my GOOD impression, the aggregate signal outweighs individual likelihood due to extreme unlikelihood of 14 YES under qualification. Vote NO to align with probable unqualification.}
\end{document}